\documentclass[11pt]{article}

\usepackage[utf8]{inputenc}
\usepackage{authblk}
\usepackage{fullpage}
\usepackage{bbm}
\usepackage{graphicx}
\usepackage{upref}
\usepackage{enumerate}
\usepackage{latexsym}
\usepackage{color,graphics}
\usepackage{comment} 
\usepackage{relsize}

\usepackage[section]{algorithm} 
\usepackage{amsmath,amsthm,amssymb,algorithmic,epsfig,graphicx, tikz}

\usepackage{amsfonts}
\usepackage{varioref}

\usepackage{amsmath}
\usepackage{amsfonts}
\usepackage{tikz} 

\usepackage{amssymb}
\usepackage{varioref}

\newtheorem{theorem}{Theorem}[section]
\newtheorem{lemma}[theorem]{Lemma}
\newtheorem{claim}[theorem]{Claim}
\newtheorem{proposition}[theorem]{Proposition}
\newtheorem{corollary}[theorem]{Corollary}

\newtheorem{defn}[theorem]{Definition}

\newcommand{\R}{\mathbb{R}}

\newcommand{\C}{\mathbb{C}}

\def\ket#1{\mathinner{|{#1}\rangle}}
\newcommand{\braket}[2]{\langle #1|#2\rangle}

\renewcommand{\part}[2]{\frac{\partial #1}{\partial #2}}

\newcommand{\all}[2]{\begin{align}\label{#2} #1\end{align}}
\newcommand{\al}[1]{\begin{align} #1\end{align}}

\newcommand{\enum}[1]{\begin{enumerate}#1\end{enumerate}}

\newcommand{\en}[1]{\left ( #1 \right )}

\newcommand{\nl}{\notag \\}

\newcommand{\norm}[1]{\lVert#1\rVert}

\long\def\rem#1{}

\newcommand{\thmref}[1]{\hyperref[#1]{{Theorem~\ref*{#1}}}}
\newcommand{\lemref}[1]{\hyperref[#1]{{Lemma~\ref*{#1}}}}
\newcommand{\remref}[1]{\hyperref[#1]{{Remark~\ref*{#1}}}}
\newcommand{\corref}[1]{\hyperref[#1]{{Corollary~\ref*{#1}}}}
\newcommand{\eqnref}[1]{\hyperref[#1]{{Equation~(\ref*{#1})}}}
\newcommand{\claimref}[1]{\hyperref[#1]{{Claim~\ref*{#1}}}}
\newcommand{\remarkref}[1]{\hyperref[#1]{{Remark~\ref*{#1}}}}
\newcommand{\propref}[1]{\hyperref[#1]{{Proposition~\ref*{#1}}}}
\newcommand{\factref}[1]{\hyperref[#1]{{Fact~\ref*{#1}}}}
\newcommand{\defref}[1]{\hyperref[#1]{{Definition~\ref*{#1}}}}
\newcommand{\exampleref}[1]{\hyperref[#1]{{Example~\ref*{#1}}}}
\newcommand{\hypref}[1]{\hyperref[#1]{{Hypothesis~\ref*{#1}}}}
\newcommand{\secref}[1]{\hyperref[#1]{{Section~\ref*{#1}}}}
\newcommand{\chapref}[1]{\hyperref[#1]{{Chapter~\ref*{#1}}}}
\newcommand{\apref}[1]{\hyperref[#1]{{Appendix~\ref*{#1}}}}

\title{Quantum machine learning with subspace states}
\author{ 
Iordanis Kerenidis \thanks{
IRIF, CNRS - University of Paris, France and  
QC Ware, Palo Alto, USA and Paris, France. 
Email: {\tt jkeren@irif.fr}.} 
\; \;  
Anupam Prakash \thanks{QC Ware, Palo Alto, USA and Paris, France. 
Email: { \tt anupamprakash1@gmail.com}.}
}

\date{\today}

\begin{document}

\maketitle 
\begin{abstract} 
We introduce a new approach for quantum linear algebra based on quantum subspace states and present three new quantum machine learning algorithms. The first is a quantum determinant sampling algorithm that samples from the distribution $\Pr[S]= det(X_{S}X_{S}^{T})$ for $|S|=d$ using $O(nd)$ gates and with circuit depth $O(d\log n)$. The state of art classical algorithm for the task requires $O(d^{3})$ operations \cite{derezinski2019minimax}. The second is a quantum singular value estimation algorithm for compound matrices $\mathcal{A}^{k}$, the speedup for this algorithm is potentially exponential. It decomposes a $\binom{n}{k}$ dimensional vector of order-$k$ correlations into a linear combination of subspace states corresponding to $k$-tuples of singular vectors of $A$. The third algorithm reduces exponentially the depth of circuits used in quantum topological data analysis from $O(n)$ to $O(\log n)$. 
%an exponential improvement in terms of circuit depth.

Our basic tool is the quantum subspace state, defined as $\ket{Col(X)} = \sum_{S\subset [n], |S|=d} det(X_{S}) \ket{S}$ for matrices $X \in \R^{n \times d}$ such that $X^{T} X = I_{d}$, that encodes $d$-dimensional subspaces of $\mathbb{R}^{n}$ and for which we develop two efficient state preparation techniques. The first using Givens circuits uses the representation of a subspace as a sequence of Givens rotations, while the second uses efficient implementations of unitaries  $\Gamma(x) = \sum_{i} x_{i} Z^{\otimes (i-1)} \otimes X \otimes I^{n-i}$ with  $O(\log n)$ depth circuits that we term Clifford loaders.

\end{abstract} 

\section{Introduction} 
Quantum machine learning is one of the most promising avenues for demonstrating quantum advantage for real-world applications. The underlying idea for quantum machine learning is quantum linear algebra, which in turn 
relies on the ability of a quantum computer to efficiently relate the eigenspaces of a unitary matrix to the corresponding eigenvalues. Given a quantum circuit corresponding to a unitary matrix $U$, the quantum {\em phase estimation} algorithm \cite{K95}, approximates the eigenvalue corresponding to each eigenspace of $U$ in coherent superposition. The phase estimation algorithm is at the core of many quantum algorithms with exponential speedups, including Shor's factoring algorithm \cite{shor1999polynomial}.

The first quantum machine learning algorithms focused on potentially exponential speedups for sparse matrices using the celebrated algorithm of Harrow, Hassidim and Lloyd \cite{HHL09} and related techniques, however these approaches had a number of caveats that make it difficult to establish end-to-end speedups \cite{A15}. One of the main challenges was that while sparse and well-conditioned matrices were well suited for the applications of the HHL algorithm, the matrices arising in machine learning were dense and often had good low-rank approximations. 

In order to address some of these challenges, quantum algorithms for preparing vector states $\ket{x} = \sum_{i} x_{i} \ket{i}$ using quantum random access memory (QRAM) data structures and algorithms for singular value estimation (SVE), an algorithmic primitive that extends phase estimation to estimating singular values of arbitrary matrices were proposed \cite{prakash2014quantum}. These results led to a quantum algorithm for recommendation systems \cite{KP16}, a strong candidate for an end-to-end speedup that overcame many of the caveats of previous approaches, namely the state preparation routines were explicitly specified, the application required only sampling from the preference matrix and the running time was polynomial (in fact sub-linear) in the rank instead of matrix dimensions. Similar QRAM based techniques have been used subsequently to obtain quantum machine learning applications, including classification, regression and clustering \cite{KL21,KLL19,KL20, KLP20, KP21}.

Subsequently, there has been a striking refutation of the potential for exponential speedups for the quantum recommendation systems and other QML algorithms based on low-rank linear algebra. 
Tang \cite{T19} obtained a quantum-inspired classical algorithm with running time polynomial in the rank and poly-logarithmic in matrix dimensions for the recommendation systems problem. This ``dequantization'' result has since been improved and extended to refute other exponential speedups for quantum machine learning problems \cite{chia2020sampling}. 
Nevertheless, the extent of quantum speedups for low-rank linear algebra is still an active area of research. The quantum algorithms for low-rank approximation and least squares regression retain an 8th-power polynomial speedup over the state of the art quantum-inspired algorithms,   \cite{gilyen2020improved}. Recent work relates the quantum-inspired approaches to state of the art leverage score sampling-based approaches in randomized numerical linear algebra, providing some evidence that a high degree polynomial speedup for the quantum algorithms may indeed survive \cite{chepurko2020quantum}.
Let us also note that the quantum-inspired algorithms are known to be considerably slower in practice than the classical algorithms with a polynomial dependence on the matrix dimension, and hence a practical quantum advantage for low-rank linear algebra, based on a theoretical high-degree polynomial speedup, remains a very viable possibility.

As discussed above, there are important bottlenecks to quantum machine learning based on both sparse and low-rank quantum linear algebra and an end-to-end exponential speedup has remained elusive for both approaches, though high-degree polynomial speedups still exist. It is therefore important to investigate if there are alternative approaches to quantum linear algebra and machine learning that mitigate some of these problems and provide both theoretical and practical quantum advantages. 

A first attempt to address the issues surrounding QML algorithms was to replace the QRAM based state preparation procedures by parametrized circuits called data loaders \cite{johri2021nearest}. A simple parameter counting argument shows that the number of qubits times the circuit depth must be $\Omega(n)$ for a data loader circuit that can construct any $n$-dimensional vector state. Near term applications with depth limited devices motivated the use of unary data loader circuits with $n$ qubits and depth $\log n$. Unary data loaders have been further used for estimating the distance or inner product between $n$-dimensional data points, and have been applied to near term similarity-based classification and clustering, for example via nearest centroid or $k$-nearest neighbors algorithms \cite{johri2021nearest} and quantum neural networks \cite{mathur2021medical}. 

The quantum advantage using these simple tools is in terms of circuit depth rather than circuit size. This implies that a single quantum processor with a native ability to perform gates on different qubits in parallel (see e.g. \cite{benchmarking2019}) can outperform a single classical processor, though one may argue that such parallelism for quantum architectures may not scale up and so the correct comparison should be against a classical parallel algorithm running on multiple classical processors, in which case only the circuit size and not depth matters. It has been an open question to find quantum speedups also in terms of circuit size using unary encodings of vectors into quantum states. 

In this paper, we present an alternative approach to quantum linear algebra and machine learning distinct from both the sparse and low-rank linear algebra approaches. The power of this approach comes from the fact that it can operate directly with subspaces of arbitrary dimension rather than just vectors (i.e. one-dimensional subspaces). Further, there are subspace analogs of most quantum machine learning algorithms including inner product estimation, singular value estimation (SVE) and transformation (SVT) with substantially larger speedups than the vector versions of these algorithms. 

%Algorithms for subspace state preparation utilize the correspondence between $n$-qubit quantum systems and systems of mutually anti-commuting operators of size $(2n+1)$ that can be formalized using Clifford algebras \cite{tsirel1987quantum}. 

At first sight, the approach of encoding $n$-dimensional vectors with $n$-qubit systems may look unpromising as it does not seem to utilize the quantum advantage offered by the exponentially large Hilbert space. However, the exponential size of the Hilbert space is instead leveraged to encode not only vectors as superpositions over bit strings of Hamming weight 1, but also all $k$-dimensional subspaces of $\R^{n}$ into superpositions over bit strings of Hamming weight $k$. The main conceptual contribution of our work is to introduce subspace states, which are quantum states encoding subspaces of arbitrary dimension. 
\begin{defn} \label{def1} 
Let $\mathcal{X}$ be a $d$-dimensional subspace of $\R^{n}$ and let $X \in \R^{n \times d}$ be a matrix with orthonormal columns ($X^{T} X = I_{d}$) such that $Col(X)= \mathcal{X}$. The subspace state $\ket{\mathcal{X} }$ is defined as, 
\al{ 
\ket{\mathcal{X}}  = \ket{ Col(X)  } = \sum_{|S|=d} det(X_{S}) \ket{ S }. \notag 
} 
where $X_{S}$ is the restriction of $X$ to the rows in $S$ and the sum is over size $d$ subsets of $[n]$. 
\end{defn} 
\noindent Subspace states represent the underlying subspace $\mathcal{X}$ and are invariant under column operations applied to the representing matrix $X$. The subspace state defined above is correctly normalized, this follows from the Cauchy Binet identity. Note also that for $d=1$, the subspace state reduces to the well-known unary amplitude encoding of a single vector.

The first question we answer in this paper is the efficient preparation of such subspace states. We present two efficient approaches to subspace state preparation in this paper. The first approach relies on the fact that any orthogonal matrix (i.e. real-valued unitary matrix) can be decomposed into a sequence of elementary Givens rotation matrices. Further, we show that given a subspace state $\ket{Col(A)}$ of arbitrary dimension, 
the state $\ket{ Col(G(i, j, \theta) A) } $  where $G(i,j, \theta)$ is a Givens rotation matrix acting on qubits $i$ and $j$, is obtained by applying a quantum gate that we call the 
$FBS$ (Fermionic Beam Splitter) gate to $\ket{Col(A)}$. 
Thus starting with the initial state $\ket{1^{k} 0^{n-k}}$, an arbitrary $k$-dimensional subspace state can be prepared by applying a suitable sequence of FBS gates. The FBS gate generalizes the $RBS$ (Reconfigurable Beam Splitter) gates previously used in data loader circuits for preparing amplitude encodings of vectors \cite{johri2021nearest}. The $FBS_{ij}(\theta)$ gate on qubits $i$ and $j$ acts as either the $RBS(\theta)$ or the $RBS(-\theta)$ gate with the sign depending on the parity of the qubits interleaved between $i$ and $j$. Note that the application of the $FBS$ gate implements the effect of a Givens rotation independent of the dimension $d$ of the subspace state, classically this operation would require $O(d)$ arithmetic operations.

Corresponding to every decomposition of an orthogonal matrix $U$ into Givens rotations, the Givens circuit $\mathcal{G}(U)$ is obtained by implementing each Givens rotation using the corresponding $FBS$ gate. We show that the subspace state corresponding to a subset $S$ of the columns of $U$ can be prepared by applying $\mathcal{G}(U)$ to $\ket{S}$. The different decompositions of the matrix $U$ into Givens rotations can have different gate complexities and depths, but they implement the same underlying unitary operation. Two such decompositions using pyramid circuits \cite{kerenidis2021classical} and the sine-cosine decomposition \cite{stewart1982computing, gawlik2018backward} are advantageous in terms of locality and circuit depth respectively. Subspace state preparation with Givens circuits is advantageous for preparing $\ket{Col(A)}$ when $A$ embeds into an orthogonal matrix $U$ with low Givens complexity, that is $U$ can be written as a product of a small number of elementary Givens rotations. The disadvantage of this approach is that for the worst case, they require a pre-processing overhead of up to $O(nd^{2})$ for computing the Givens decomposition.

The second approach to subspace state preparation mitigates this pre-processing overhead by introducing Clifford loader circuits $\mathcal{C}(x)$. Clifford loader circuits can be defined using unary data loaders. A unary data loader $D(x)$ for vector $x \in \R^{n}$ is a quantum circuit implemented with RBS gates such that $D(x) \ket{ e_1 } = \sum_{i} x_{i} \ket{ e_i}$, where $e_i$ is the $n$-bit string $0^{i-1}10^{n-i}$. Clifford loader circuits are defined in terms of the data loader circuits $D_{F}(x)$ where the $RBS$ gates have been replaced by the corresponding $FBS$ gates. 
\begin{defn} \label{cloader} 
Given a unary data loader $D(x)$, the Clifford loader is the circuit $\mathcal{C}(x) = D_F(x)(X\otimes I^{\otimes n-1}) D_F(x)^{*}$ where $D_F(x)$ is obtained by replacing each $RBS$ gate in $D(x)$ by the corresponding $FBS$ gate. 
\end{defn} 
\noindent The main results on Clifford loaders show that the Clifford loader circuit $\mathcal{C}(x)$ is independent of the unary data loader circuit $D(x)$ and always implements the unitary operator, 
\al{ 
\Gamma(x) = \sum_{i} x_{i} Z^{\otimes (i-1)} \otimes X \otimes I^{\otimes (n-i)} 
} 
That is $\mathcal{C}(x)$ encodes the vector $x$ as a sum of the mutually anti-commuting operators generating the Clifford algebra.
The action of the Clifford loader circuit on the entire Hilbert space can be described in terms of subspace states, if $\ket{Col(Y)}$ contains the vector $x$ then the Clifford loader creates the subspace state $\ket{Col(Y')}$ such that $Col(Y)= (Y', x)$ and if $\ket{Col(Y)}$ is orthogonal to $x$, then it creates the state $\ket{Col(Y, x)}$ by adding $x$
to the subspace $Col(Y)$. Thus, a subspace state $\ket{Col(X)}$ can be prepared by applying successively the Clifford loader circuits corresponding to the columns of the representing matrix. 

The gate complexity, circuit depth, and pre-processing overheads for the different subspace state preparation methods 
are summarized in Table \ref{tab1}.

\begin{table} \label{tab1} 
\begin{center} 
  \begin{tabular}{|l|c|c|c|} 
   \hline
  &&&\\
   \textbf{Method} & \textbf{Gates} & \textbf{Depth} & \textbf{Preprocessing}  \\
      &&&\\
    \hline
      &&&\\
      Pyramid Circuit & $O(nd)$ & $n+d$ & $O(nd^{2})$ \\
    &&&\\
    \hline
     &&&\\
    Sine Cosine Decomposition   & $O(nd)$ & $d \log n$ & $O(nd^{2})$ \\
      &&&\\
    \hline 
     &&&\\
      Clifford Loaders  & $O(nd)$ & $d \log n$ & $O(nd)$ \\
    &&&\\
    \hline  
   \end{tabular} 
   \caption {Gate complexity, circuit depth, and pre-processing overheads for subspace state preparation methods.  } \label{table1} 
  \end{center} 
   \end{table}

Using the subspace states and the techniques developed for preparing them, we present three new quantum machine learning algorithms. The first algorithm is for determinant sampling, a task with several applications to numerical linear algebra and machine learning \cite{derezinski2021determinantal}. Given matrix $A \in \R^{n \times d}$, the determinant sampling problem is to sample subsets $S$ of the rows with $|S|=d$ such that $\Pr[S]= \frac{\det(A_{S}^{2})}{\det(A^{T}A)}$. 

It follows from the definition of subspace states that determinant sampling is equivalent to preparing the subspace state $\ket{Col(A)}$ and measuring in the standard basis. As the determinant sampling distribution depends only on the column space, classical algorithms typically pre-process the matrix to have orthonormal columns. The best known classical algorithm for determinant sampling requires $O(d^{3})$ arithmetic operations for generating a single sample from the determinant distribution after pre-processing the matrix to have orthonormal columns which takes $O(nd^{2})$ operations  \cite{derezinski2019minimax}. 

The two approaches to subspace state preparation in this work can be used to obtain algorithms for quantum determinant sampling. The first step for the quantum determinant sampling algorithm is the same as in the classical algorithms, namely to pre-process the matrix to have orthonormal columns which takes $O(nd^{2})$ operations. Then, the Clifford loader circuits provide an algorithm for generating samples from the determinant distribution using a quantum circuit with $O(nd)$ gates and depth $O(d \log n)$. As the above informal description of the action of the Clifford loader indicates, applying sequentially the Clifford loader for the columns of a matrix $A$ with orthonormal columns suffices to create the subspace state $\ket{Col(A)}$. An alternative proof of this result is also given using a path integral style argument to compute explicitly the amplitudes for the composition of Clifford loaders and show that they equal the corresponding determinants. 

\begin{theorem} 
Given $A \in \R^{n \times d}$, there is a quantum determinant sampling algorithm that requires $O(nd^{2})$ pre-processing to orthogonalize $A$, and generates subsequent samples 
using a quantum circuit with $O(nd \log n)$ gates and with total depth $O(d \log n)$. 
\end{theorem}

\noindent For the case $d=O(n)$, the quantum determinant sampling algorithm has gate complexity $O(n^{2})$ compared to the state-of-the-art classical algorithm with gate complexity $O(n^{3})$. Note also that the quantum circuit is also quite shallow with depth only $O(n \log n)$.
The source of this speedup is the implementation of a classical row operation of complexity $O(n)$ by a single quantum $FBS$ gate. We also identify a family of matrices where this speedup can be cubic, these are matrices embedding into unitary matrices that can be written as a product of $O(n \log n)$ elementary Givens rotations. This class of matrices is known to include Fourier, wavelet and related transforms  \cite{rusu2021fast} and has recently attracted attention in the machine learning literature \cite{frerix2019approximating}. For this case, the additional speedup is due to the low Givens complexity of the matrix, while it remains to be seen if more efficient classical determinant sampling algorithms can be given for this setting.

Although the speedups obtained for quantum determinant sampling are polynomial, the tools developed for these speedups namely Givens circuits and Clifford loaders can also be used to obtain higher degree polynomial and 
potentially exponential quantum speedups. We provide such an application for both of these tools.

The second main result of this work is a singular value decomposition algorithm for compound matrices using Givens circuits. Compound matrices represent the action of a matrix on subspaces rather than vectors and are defined as follows. 
\begin{defn} \label{def:cmpd} 
[Compound matrix] Given a matrix $A \in \R^{n \times n}$, the compound matrix $\mathcal{A}^{k}$ for $k \in [n]$ is the $\binom{n}{k}$ dimensional matrix with entries 
$\mathcal{A}^{k}_{IJ} = det(A_{IJ})$ where $I$ and $J$ are subsets of rows and columns of $A$ with size $k$. 
\end{defn} 
\noindent The compound matrices are of exponential size for large enough $k$, however it is possible to perform $\ell_{2}$-sampling on the rows and columns of compound matrices using classical determinant sampling algorithms. The quantum inspired algorithms that have been successful in ``dequantizing'' singular value transformations for low rank matrices require $\ell_{2}$-sampling access, the high computational cost of $\ell_{2}$ sampling for compound matrices makes quantum SVD methods for these matrices more resistant to dequantization methods. Also, if the matrix $A$ has a good rank-$r$ approximation, the rank of the compound matrices $\mathcal{A}^{k}$ is exponential in $r$ for large $k$, 
and thus a dequantization method that relies on the existence of a low rank approximation would not be applicable.  It is therefore plausible that the quantum algorithm indeed achieves an exponential speedup for some settings, but further work on quantum inspired SVD for compound matrices would be needed to substantiate the claim of an exponential speedup. An equally important question is whether this speedup can be shown in practice, whether exponential or high-degree polynomial.

The singular vectors for the $k$-th order compound matrix $\mathcal{A}^{k}$ are the subspace states corresponding to subspaces spanned by $k$ tuples of singular vectors of $A$. Further, a block encoding for $A$ in orthogonal matrix $U$ extends to a block encoding of $\mathcal{A}^{k}$ in $\mathcal{U}^{k}$ and this can be used to derive the quantum singular value estimation algorithm for $\mathcal{A}^{k}$. In effect, the quantum SVD for compound matrices 
$\mathcal{A}^{k}$ takes a vector of order $k$ correlations and decomposes it as a linear combination of correlations arising from subspaces of singular vectors of $A$. 

\begin{theorem} 
(Subspace SVE) Let $A=U\Sigma V^{T}$ be the singular value decomposition for matrix $A \in \R^{m \times n}$ embedding in unitary $U \in \R^{N \times N}$. Then the mapping, 
\al{ 
\ket{\phi_{k}} := \sum_{S \in \mathcal{H}_{k}} \alpha_{S} \ket{Col(V_{S}) } \to  \sum_{S \in \mathcal{H}_{k}} \alpha_{S} \ket{Col(V_{S}) } \ket{ \overline{\theta_{S}} }
} 
such that $|\cos(\overline{\theta_{S}}) - \prod_{i \in S} \sigma_{i} | = O(\epsilon)$ has complexity $O(T(U)/\epsilon)$ where $T(U)$ is the complexity of implementing the Givens circuits for $U$. 
\end{theorem} 
\noindent Similar to the application of the SVE algorithm to recommender systems, the subspace SVE algorithm can have applications to data analysis and for collaborative filtering style recommender algorithms where recommendations are made jointly to groups of users.

The third main result of this work is an exponential depth reduction for quantum topological data analysis, the depth of the necessary quantum circuit is reduced from $O(n)$ to $O(\log n)$ using the logarithmic depth Clifford loader circuit.
Topological data analysis is based on spectral analysis of the Dirac operator for a simplicial complex, which is an exponentially large but sparse matrix. The main observation for the depth reduction result is that the 
Dirac operator for an arbitrary simplicial complex embeds into the Clifford loader unitary $\Gamma(x)$ for the uniform vector. 
\begin{theorem} 
The Dirac operator of an orientable simplicial complex $C$ is the submatrix of  $U= \frac{1}{\sqrt{n}} \sum_{i \in [n]} Z^{\otimes (i-1)}\otimes  X \otimes I^{\otimes (n-1)}$ indexed by the simplices that belong to $C$. 
\end{theorem} 
\noindent The efficient logarithmic depth construction of the corresponding Clifford loader circuit therefore 
also provides a succinct block encoding for the Dirac operator improving exponentially on previous constructions. We note that recently there have been a series of works \cite{ubaru2021quantum, gyurik2020towards, hayakawa2021quantum} on quantum TDA that improve substantially on the original LGZ algorithm \cite{lloyd2016quantum}, by constructing more efficient encodings for the Dirac operator. Using Clifford loaders to implement the Dirac operator provides an exponential imrpovement on the circuit depth on all these constructions as well.

We have provided three applications of quantum subspace states and the circuits used for preparing them. It is likely that further applications will be found and that these tools will be broadly useful for obtaining polynomial 
and possibly even exponential speedups for machine learning and topological data analysis problems. 
The paper is organized as follows. In section 2, we introduce subspace states and Fermionic beam splitter gates ($FBS$) and demonstrate that these gates applied to subspace states implement Givens rotations irrespective of the subspace dimension. 
This yields subspace state preparation algorithms using Givens circuits in section 3. In section 4, we introduce Clifford loaders and an alternative subspace state preparation algorithm. In section 5, we present the results on quantum determinant sampling. Section 6 describes the 
quantum SVD algorithm for compound matrices using Givens circuits. Finally in section 7, we describe the exponential depth reduction for topological data analysis using Clifford loaders.

\section{Subspace states} 
\noindent This section introduces subspace states as higher dimension generalizations of the unary encodings of vectors. Subspace states encode subspaces of $\R^{n}$ of arbitrary dimension and are invariant under unitary operations applied to the subspaces. We start by providing a high level overview of subspace states.  

Let $H$ be the $2^{n}$-dimensional Hilbert space for an $n$-qubit quantum system. The standard basis for $H$ is indexed by bit strings in the $n$-dimensional hypercube $\mathcal{H} = \{ 0, 1\}^{n}$. A quantum state $\ket{\phi} \in H$ can be written as $\ket{\phi} = \sum_{d \in [n]} \ket{\phi_{d} }$ where $\ket{\phi_{d}}$ is a superposition over bit strings of Hamming weight $d$. We denote the Hilbert space spanned by Hamming weight $d$ strings as $H_{d}$, and note that $H_{d}$ has dimension $\binom{n}{d}$. 

It is well known that there is a correspondence between unit vectors in $\R^{n}$ and quantum states in $H_{1}$ given by the unary amplitude encodings of the vectors, the $i$-th coordinate of the vector becomes the amplitude of the unary basis state $\ket{e_i}$ in the unary encoding. Further, one can map the state $\ket{e_{1}}$ into the unary encoding of an arbitrary vector $x$ by applying a suitable sequence of two-qubit gates. 
Such a construction follows from the fact that for all vectors $x$, there exists a unitary matrix $U$ such that $Ue_{1} = x$ and all unitaries can be decomposed as a product of elementary Givens rotation matrices
that act non trivially on only two rows and columns. Each Givens rotation can then be applied through a two-qubit gate called a Reconfigurable Beam Splitter ($RBS$) gate.

In this section, we extend this correspondence to an embedding of $d$-dimensional subspaces of $\R^{n}$ into the Hilbert space $H_{d}$. For this, we introduce subspace states that are quantum states in $H_{d}$ corresponding to $d$-dimensional subspaces of $\R^{n}$. Further, we show that for any subspace state $\ket{\phi}$ there exists a sequence of unitary operations that each modifies only two qubits and that map the state $\ket{1^{d} 0^{n-d}}$ to $\ket{\phi}$. 
Similar to the case of unary amplitude encodings, this follows from the fact that for an arbitrary $d$-dimensional subspace with orthonormal basis $(x_{1}, x_{2}, \cdots, x_{d})$, there exists a unitary matrix $U$ such that $Ue_{i} = x_{i}$ for $i \in [d]$ and that the unitary $U$ can be decomposed as a product of elementary Givens rotation matrices. We define the Fermionic Beam Splitter ($FBS$) gate, a generalization of the $RBS$ gate, that implements the effect of a Givens rotation on a subspace state. Interestingly, an FBS gate implements the effect of a Givens rotation independent of the dimension of the subspace state. Replacing the Givens rotation matrices by the corresponding $FBS$ gates yields a quantum circuit for subspace state preparation. 

The next subsection recalls basic results on the correspondence between $H_{1}$ and vectors in $\R^{n}$ and the action of the $RBS$ gates, and subsequently these results are generalized to higher dimensional subspace states and the action of the $FBS$ gates.

\subsection{Unary amplitude encodings} 
Bit strings with Hamming weight $1$ are commonly used to encode $n$-dimensional vectors into quantum states in quantum machine learning \cite{johri2021nearest}. 
We recall the definition of the unary amplitude encoding of a unit vector. 
\begin{defn} \label{def1} 
Let $x \in \R^{n}$ be a unit vector, the unary amplitude encoding $\ket{ un(x) } $ is defined as, 
\al{ 
\ket{ un(x) } = \sum_{i \in [n]} x_{i}  \ket{ e_i }, 
} 
where $e_i$ is the string $0^{i-1}10^{n-i}$.
\end{defn} 
\noindent Unary amplitude encodings are useful for quantum machine learning as they can be prepared using parametrized circuits of depth $\log n$, in contrast to binary encodings that either require depth $O(n)$ or specialized hardware like quantum random access memory (QRAM) for state preparation. The state preparation circuits for preparing unary amplitude encodings are known as unary data loaders  \cite{johri2021nearest}. 

We next define the action of the orthogonal group on the unary amplitude encodings. Note in this work we restrict ourselves to quantum states with real amplitudes and therefore consider the orthogonal group instead of the unitary group. Let us start by defining the elementary Givens rotation matrices that act non trivially on two dimensional subspaces. 
\begin{defn} 
The Givens rotation matrix $G(i, j, \theta)$ in $\R^{n\times n}$ for $i>j$ has non zero entries given by $G_{kk} = 1 $ if $k\neq i, j$, $G_{ii} = G_{jj} = \cos(\theta)$ and $G_{ij} = -G_{ji} = \sin(\theta)$.  
\end{defn} 
\noindent An orthogonal matrix can be decomposed into a sequence of elementary Givens rotation matrices. The action of a single Givens rotation matrix on a unary amplitude encoding is given exactly by a two-qubit gate called the Reconfigurable Beam Splitter ($RBS$) gate. 

An $RBS$ gate is a two-qubit gate parametrized by a single angle $\theta$. Letting the basis states in the standard basis be ordered as $\ket{00}, \ket{01}, \ket{10}, \ket{11}$, the $RBS$ gate has the following representation in the 
standard basis, 
\begin{equation} 
RBS(\theta) =  \begin{pmatrix} 
&1  &0 & 0 & 0 \\  
&0  & \cos(\theta) &\sin(\theta)  &0 \\
&0 & -\sin(\theta)   & \cos(\theta)   & 0 \\
&0 & 0 & 0     & 1 
\end{pmatrix}. 
\label{btheta} 
\end{equation}  
Note that the conjugate of the $RBS(\theta)$ gate is obtained by reversing the angle, that is $RBS(\theta)^{*}= RBS(-\theta)$. The notation $RBS_{ij} (\theta)$ denotes the $RBS(\theta)$ gate applied to qubits $i$ and $j$. 
The $RBS_{ij} (\theta)$ gate implements the action of a Givens rotation matrix on a 
unary amplitude encoding of a unit vector $x \in \R^{n}$, that is
\all{ 
RBS_{ij} (\theta)\ket{  un(x) } = \ket{ un(G(i, j, \theta) x) }. 
} {eq:givens} 
The logarithmic depth construction of the unary data loader (Figure \ref{fig3}) uses $RBS_{ij} (\theta)$ gates and can be viewed as an application of the above equation showing that a single $RBS_{ij} (\theta)$ gate implements a Givens rotation matrix. 

\begin{figure}[H] 
	\centering
	\includegraphics[scale = 0.6]{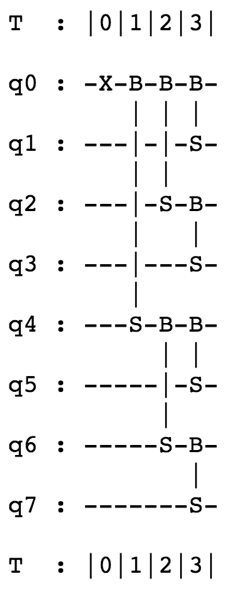}
	\caption{Quantum circuit corresponding to the unary data loader implemented using RBS gates on 8 qubits. } \label{fig3} 
\end{figure}

The $RBS$ gate is an example of a matchgate which is a two qubit gate that decomposes into rotations on the subspaces spanned by $\ket{01}, \ket{10}$ and $\ket{00}, \ket{11}$. Quantum circuits consisting only of matchgates can be simulated classically with a polynomial overhead \cite{terhal2002classical} if the matchgates act on adjacent qubits. Matchgates with longer range interactions are in fact universal for quantum computing \cite{brod2012geometries}. 
For quantum linear algebra algorithms using $O(n)$ qubits to operate with $n$ dimensional matrices, the efficient simulability of the circuits is not a bottleneck for obtaining polynomial speedups. The algorithms achieving potential exponential speedups apply controlled versions of circuits composed of $FBS$ gates and therefore use gates besides the ones that can be efficiently simulated.

\subsection{Subspace states} 
This section defines quantum subspace states that generalize the  unary amplitude encodings of a vector $x \in \R^{n}$ to quantum states encoding $d$-dimensional subspaces of $\R^{n}$ 
for $d \in [n]$. 
A vector $x \in \R^{n}$ defines a one-dimensional subspace of $\R^{n}$. More generally, a $d$-dimensional subspace is specified by an orthonormal basis $(x^{1}, x^{2}, \cdots, x^{d})$ where each $x^{i} \in \R^{n}$ is a unit vector 
and $\braket{ x^{i} } { x^{j}} =\delta_{ij}$. 

\begin{defn} \label{def2} 
Let $\mathcal{X}$ be a $d$-dimensional subspace of $\R^{n}$ for $d \in [n]$ and let $X \in \R^{n \times d}$ be a matrix with orthonormal columns ($X^{T} X = I_{d}$) such that $Col(X)= \mathcal{X}$. The subspace state $\ket{\mathcal{X} }$ is defined as, 
\al{ 
\ket{\mathcal{X}}  = \ket{ Col(X)  } = \sum_{|S|=d} det(X_{S}) \ket{ S }.
} 
where $X_{S}$ is the restriction of $X$ to the rows in $S$ and the sum is over size $d$ subsets of $[n]$. 
\end{defn} 
Some remarks on this definition follow. First, there is a sign ambiguity in the representation of vectors and subspaces as quantum states. In the case of vectors $\ket{x}$ and $\ket{-x}$ correspond to the same quantum state. 
In the case of subspaces, if any two columns of the representing matrix are swapped all the determinants $X_{S}$ change in sign but the subspace state remains invariant. The subspace state remains invariant under sign flips and  permutations applied to the columns of the representing matrix $X$ as these operations preserve the underlying subspace. 

In fact, definition \ref{def2} implies that a subspace state is independent of the choice of the orthonormal basis for $\mathcal{X}$ as the determinants $det(X_{S})$ are invariant under column operations and two different orthonormal bases for $\mathcal{X}$ can be mapped into one another with a suitable sequence of column operations. The subspace state depends only on the subspace $\mathcal{X}$ and not on the particular orthonormal basis used to represent it. Finally a remark on the notation. As observed above, the subspace state can be defined intrinsically in terms of the subspace $\ket{\mathcal{X}}$ and also in terms of the representing matrix $\ket{Col(X)}$. Both these notations are used and can be interpreted from the context.

The Cauchy Binet identity stated below shows that that the subspace state is a correctly normalized quantum state if and only if $(x^{1}, x^{2}, \cdots, x^{d})$ are an orthonormal basis for $\mathcal{X}$, the orthonormality ensures that 
$\sum_{|S|=k} det(X_{S})^{2} = det(X^{T} X)=1$. 

\begin{claim} \label{cbi} 
(Cauchy-Binet Identity) Let $X\in \R^{n \times d}$ and $Y \in R^{n \times d}$ then, 
$$det(X^{T}Y) = \sum_{|S|=d} det(X_{S}) det(Y_{S}) $$. 
\end{claim} 
\noindent The Cauchy Binet identity also can be used to compute the inner product between subspace states $\braket{\mathcal{X}} { \mathcal{Y} } = det(X^{T} Y)$. The inner product can be interpreted as the product of the cosines of the principal angles between the subspaces.

Subspace states of dimension $d$ are determined by exactly $O(nd)$ parameters and thus occupy a small fraction of the $\binom{n}{d}$-dimensional subspace $\mathcal{H}_{d}$ spanned by the bit strings 
of Hamming weight $d$. Each standard basis state $\ket{S}$ with $Ham(S)=d$ is a subspace state corresponding to the matrix $X$ with $X_{s(i) i} =1$ where $s(i), i \in [d]$ are the non zero elements of $S$. 
Unlike vector states, the linear combination of subspace states is not a a subspace state, the amplitudes of subspace states are in fact characterized by quadratic constraints. 
 These constraints are called the Van Der Waerden syzigies, and they provide a complete set of quadratic constraints specifying the algebraic variety of the subspace states \cite{sturmfels2008algorithms}.

Subspace states are invariant under column operations on $X$, however row operations change the column space and also the subspace state. We now want to show that an arbitrary $d$-dimensional subspace state can be generated starting from the state
$\ket{1^{d} 0^{n-d}}$ and applying a suitable sequence of quantum gates. The building block for the construction is the implementation of the Givens rotation $G(i, j, \theta)$ on subspace states using an $FBS$ gate defined in the next section. This can be viewed as a generalization of equation \eqref{eq:givens} which showed that $RBS$ gates implement the Givens rotation on unary encodings of vectors.

\subsection{Fermionic Beam Splitter gates} 

In this section we show that the orthogonal group has a natural action on the subspace states for all $d \in [n]$. We introduce a family of gates that we call the Fermionic Beam Splitter ($FBS$) gates 
and show that they implement the action of Givens rotation matrices on subspace states generalizing equation \eqref{eq:givens}.

The definition of a Fermionic gate assumes an ordering on the qubits, in particular the action of the $FBS_{ij}(\theta)$ gate acting on qubits $i$ and $j$ depends on the parity of the qubits between $i$ and $j$. 
%Namely, the $FBS_{ij}(\theta)$ gate acts either as $RBS_{ij}(\theta)$ or $RBS_{ij}(-\theta)$ depending on the parity of the intermediate qubits beings 0 or 1. 

\begin{defn} \label{def:fbs} 
Let $i, j \in [n]$ be qubits and let $f(i, j, S)= \sum_{i< k <j } s_{k}$ for $S \in \{ 0, 1 \}^{n}$ where $n$ is the total number of qubits. Then the FBS gate acts on qubits $i$ and $j$ as the unitary matrix below, 
\begin{equation} 
FBS_{ij} (\theta) \ket{S} =  \begin{pmatrix} 
&1  &0 & 0 & 0 \\  
&0  & \cos(\theta) &(-1)^{f(i, j, S)}\sin(\theta)  &0 \\
&0 & (-1)^{f(i, j, S)+1}  \sin(\theta)   & \cos(\theta)   & 0 \\
&0 & 0 & 0  & 1 
\end{pmatrix} 
\label{fbs} 
\end{equation} 
The gate acts as identity on the remaining qubits. 
\end{defn} 
\noindent It follows from the definition that $FBS_{ij} (\theta)$ acts as $RBS_{ij} (\theta)$  if the parity $\oplus_{i<k<j} x_{k}$ of the qubits 'between' $i$ and $j$ is even and is the conjugate gate $RBS_{ij}(-\theta)$ otherwise. 
Also note that the $FBS$ gate is identical to the $RBS$ gate if $i$ and $j$ are consecutive qubits, meaning there are no qubits between $i$ and $j$. It is also identical to the $RBS$ gate on $\mathcal{H}_{1}$
as $f(i,j,S)=0$ if the only $1$ in the bit string is present at index $i$ or $j$. This is also why for the unary data loader we can simply use the $RBS$ gates.

In the general setting, the $FBS_{ij} (\theta)$ is a non local gate that can can be implemented using an $RBS$ gate together with $O(|i-j|)$ additional two qubit parity gates with a circuit of depth $O(\log |i-j|)$. The next claim provides such an efficient implementation for the $FBS_{ij} ( \theta)$ gate. 

\begin{proposition} \label{fbsimp}
The gate $FBS_{ij} (\theta)$ can be implemented using an RBS gate together with a circuit of depth $O(\log |i-j|)$ with $O( |i-j|)$ additional two qubit controlled-$X$ and controlled-$Z$ gates. 
\end{proposition} 
\begin{proof} 
Consider first the case $|i-j|=2$, without loss of generality assume that the $FBS$ gate depends on three qubits in ascending order, these three qubits can be labeled $1,2$ and $3$ for convenience. The $FBS_{13} (\theta)$ gate acts as follows: if qubit 2 is 0 the gate $RBS_{13}(\theta)$ is applied to qubits 1 and 3 and if qubit 2 is 1 it the gate $RBS_{13}(-\theta)$ is applied instead. This implies the following decomposition
as a sequence of three two-qubit gates, 
\all{ 
FBS_{13} ( \theta)(\theta) = CZ_{21} RBS_{13}(\theta) CZ_{21}
} {eq:3} 
where the $CZ_{21}$ is the controlled-$Z$ gate with qubit $2$ acting as the control gate. If qubit $2$ is $0$ then the $CZ_{21}$ gates act as identity and $CZ_{21} RBS_{13}(\theta) CZ_{21} = RBS_{13}(\theta)$ matching the action 
of $FBS_{13} (\theta)$. If qubit 2 is $1$ then the identity $CZ_{21} RBS_{13}(\theta) CZ_{21} = (Z \otimes I)  RBS_{13}(\theta)  (Z \otimes I) = RBS_{13}(-\theta)$ holds, establishing the correctness of equation \eqref{eq:3}. 

The implementation of $FBS_{ij} ( \theta)$ can be reduced to the case $|i-j|=2$ by computing first the parity of the qubits between $i$ and $j$ in qubit $i+1$ and then using the three gate sequence in equation \eqref{eq:3} and undoing the 
parity computations. 
For $|i-j|=k+1$, let $P(i+1, j)$ be the circuit with $(k-1)$ CX gates and depth $\lceil \log k \rceil $ that computes the parity of the $k$ qubits between $i$ and $j$ on qubit $(i+1)$, then it follows that $FBS_{ij} (\theta)(\theta) = P(i+1, j) CZ_{i+1,i} RBS_{ij}(\theta) CZ_{i+1, i} P^{\dagger}(i+1, j)$. 

\end{proof} 

\subsection{Quantum Givens rotations} 
The next result shows that an $FBS$ gate acting on subspace states implements the action of the Givens rotation matrix. Interestingly, the result does not depend on the 
dimension of the subspace state, so the application of a single $FBS$ gate can simulate the effect of a row operation on an $n \times d$ matrix, an operation with complexity $O(d)$ in the classical case. 

A state $\ket{\phi} = \sum_{k} \ket{ \phi_{k}}$ in the $n$ qubit Hilbert space $H$ can be regarded as a linear combination of Hamming weight $k$ dimensional components  $\ket{\phi_{k}} \in H_{k}$. The $FBS$ gate is Hamming weight preserving, generalizing equation \eqref{eq:givens}, we show that 
the application the $FBS$ gate implements the effect of a  Givens rotation on all the components $\ket{ \phi_{k}}$. 
Note that in general $\ket{ \phi_{k}}$ is a linear combination of subspace states, the effect of the $FBS$ gate is to implement the effect of a Givens rotation on each component of the linear combination.

\begin{theorem} \label{givens} 
An application of the gate $FBS_{ij} ( \theta)$ applies a Givens rotation to the subspace state, that is, 
\begin{equation} 
FBS_{ij} ( \theta) \ket{Col(X)}  = \ket{ Col(G(i, j, \theta) X) }. 
\end{equation} 
where $X \in\R^{n \times d} $ is the matrix representing an arbitrary subspace state of dimension $d \in [n]$. 
\end{theorem} 
\begin{proof} 
The proof proceeds by computing explicitly the amplitudes for the states $FBS_{ij} (\theta)\ket{Col(X)}$ and  $\ket{ Col(G(i, j, \theta) X) }$ and demonstrating equality. 
The $FBS_{ij} ( \theta) $ is a Hamming weight preserving gate and $\ket{ Col(X) }$ is a superposition over bit strings of Hamming weight $d$, so it follows that $FBS_{ij} (\theta)\ket{Col(X) } = \sum_{ |S|=d} g(S) \ket{S}$ for some amplitudes $g(S)$. The 
amplitudes $g(S)$ for all $|S|=d$ can be computed from definition \ref{def:fbs} as,  
\al{ 
g(S)= \begin{cases} 
det(X_{S}) \;\; \; \; &\text{if $i, j \not \in S$ or if $i, j \in S$} \nl 
\cos(\theta) det(X_{S})  - (-1)^{f(S)} \sin(\theta) det(X_{S\setminus j \cup i}) \;\; \; \; &\text{if $i \not \in S$ and $ j \in S$} \nl 
\cos(\theta) det(X_{S})  + (-1)^{f(S)} \sin(\theta) det(X_{S\setminus i \cup j})  \;\; \; \; &\text{if $j \not \in S$ and $i \in S$} 
\end{cases} 
} 
In the above equations, $det(X_{S\setminus i \cup j} )$ is the determinant of the 
matrix where the rows indexed by $S\setminus i \cup j$ are sorted in ascending order.

Computing the amplitudes for $\ket{ Col(G(i, j, \theta) X) }$, which is the subspace state corresponding to the matrix with rows $(x_{1}, \cdots, x_{i-1}, \cos(\theta) x_{i} + \sin(\theta) x_{j}, x_{i+1} \cdots, x_{j-1}, 
\cos(\theta) x_{j} - \sin(\theta) x_{i}, x_{j+1}, \cdots, x_{n})$. The amplitudes for this state are given by,  
\al{ 
h(S)= \begin{cases} 
det(X_{S}) \;\; \; \; &\text{if $i, j \not \in S$ or if $i, j \in S$} \nl 
\cos(\theta) det(X_{S})  -  \sin(\theta) det(X_{S: x_{j} \to x_{i} })  \;\; \; \; &\text{if $i \not \in S$ and $ j \in S$} \nl 
\cos(\theta) det(X_{S})  +  \sin(\theta) det(X_{S: x_{i} \to x_{j} })   \;\; \; \; &\text{if $j \not \in S$ and $i \in S$} 
\end{cases} 
} 
Note that $det(X_{S: x_{i} \to x_{j} })$ is the determinant of the matrix where $x_{i}$ has been replaced by $x_{j}$.  The number of row swaps needed to transform  $X_{S: x_{i} \to x_{j} }$ to $X_{S\setminus i \cup j}$ 
(and also $X_{S: x_{j} \to x_{i} }$ to $X_{S\setminus j \cup i}$) is equal to the number of elements of $S$ that lie between $i$ and $j$, thus, 
\al{ 
 det(X_{S: x_{i} \to x_{j} }) = (-1)^{f(S)} det(X_{S\setminus i \cup j}) \nl
  det(X_{S: x_{j} \to x_{i} }) = (-1)^{f(S)} det(X_{S\setminus j \cup i}). 
} 
It follows that the amplitudes $g(S)$ and $h(S)$ are equal for all $S$ such that $|S|=d$ and thus $FBS_{ij} ( \theta) \ket{Col(X)}  = \ket{ Col(G(i, j, \theta) X) }$ for all $d\in [n]$.

\end{proof}

\section{Subspace state preparation with Givens circuits} \label{ssp} 

We are now ready to describe the first approach for preparing the subspace states. Note that a $d$-dimensional subspace state has up to $\binom{n}{d}$ non zero amplitudes that are specified as the determinants of sub-matrices of the representing matrix. Subspace state preparation is a non trivial problem as it creates an exponentially large superposition, where each amplitude requires time $O(d^{\omega})$ to compute classically. The input for the subspace state preparation problem is a matrix $X \in \R^{n \times d}$ with orthonormal columns and the desired output is the state $\ket{Col(X)}$. In the following sections we present efficient quantum circuits for subspace state preparation.

Theorem \ref{givens} suggests an algorithm for the subspace state preparation problem using Givens decompositions of unitary matrices. Let $I_{n,d}$ be the matrix consisting of the first $d$ columns of the $n \times n$ identity matrix, the corresponding subspace state is $\ket{1^{d} 0^{n-d}}$. The first step is to find a decomposition of the matrix $X$ as a product of Givens rotation matrices, that is a decomposition of the form $X = \prod_{i, j \in [n]} G(i, j, \theta_{ij}) I_{n,d}$. The subspace state $\ket{\mathcal{X}}$ can then be constructed by applying sequentially the $FBS_{ij} (\theta)$ gates corresponding to the Givens rotation matrices starting from the state $\ket{1^{d} 0^{n-d}}$.

In this section we show that this approach indeed suffices for the construction of subspace states.  
We define a class of quantum circuits consisting of $FBS$ gates and prove a general result by describing the action of such circuits on the standard basis for the $n$ qubit Hilbert space.

\begin{defn} \label{def:U} 
A Givens circuit $\mathcal{G}= \prod_{(i, j, \theta) \in \mathcal{T} } FBS_{ij} (\theta)$ is the product of $FBS_{ij} (\theta)$ gates for some set of indices $(i, j, \theta) \in \mathcal{T}$. The unitary $U(\mathcal{G}) =  \prod_{(i, j, \theta) \in \mathcal{T} } G(i, j, \theta)$ corresponding to a Givens circuit is the product of the corresponding Givens rotation matrices. 
\end{defn} 

\noindent The next theorem describes the action of a Givens circuit on the standard basis for a Hilbert space of $n$ qubits. It shows that the Givens circuit for a unitary $U$ acting on standard basis state $\ket{S}$ prepares the subspace state corresponding to the matrix $U_{S}$, obtained by selecting the columns of $U$ indexed by $S$. The Givens circuit can therefore be used to prepare subspace states corresponding to any subset of the columns of $U$. 
\begin{theorem} \label{thm: action} 
Let $\mathcal{G}$ be a Givens circuit on $n$ qubits and let $U= U(\mathcal{G})$ be the corresponding unitary, then the action of $\mathcal{G}$ on a standard basis state $\ket{S}$ for $S\in \{0, 1\}^{n}$ is given as, 
\al{ 
\mathcal{G} \ket{S} =\ket{   Col(U_{S} ) }, 
} 
that is $\mathcal{G}$ maps the standard basis state $\ket{S}$ to the subspace state for matrix $U_{S}\in \R^{n\times |S|}$, the submatrix of $U$ with columns indexed by $S$. 
\end{theorem} 
\begin{proof} 
For $S= \ket{0^{n}}$ or  $S= \ket{1^{n}}$, the claim holds as these states are invariant under $FBS_{ij} (\theta)$ for all $i, j \in [n]$. 
It follows that $\mathcal{G} \ket{e_i}= \ket{un(u_{i})}$ since applying the sequence of Givens rotations corresponding 
to $\mathcal{G}$ on  $\ket{e_i}$ creates the state $\ket{un(u_{i})}$ corresponding to the unary encoding of the $i$-th column of $U$ by definition \ref{def:U}.  

The state $\ket{S}$ can be viewed as a subspace state corresponding to the matrix $I_{S}$ obtained by restricting the identity matrix $I \in \R^{n \times n}$ to the columns in $S$. 
Applying Theorem \ref{givens} to the circuit $\mathcal{G}$ viewed as a sequence of $FBS$ gates maps $\ket{I_{S}}$ to $\ket{Col(X)}$ for the matrix $X$ obtained by applying the corresponding Givens rotations to $I_{S}$. 
Definition  \ref{def:U} of the unitary $U(\mathcal{G})$ corresponding to the Givens circuit implies that this subspace state is $\ket{  Col( U_{S} ) }$. 
\end{proof} 
In order to prepare $\ket{Col(X)}$ it suffices to embed $X$ into a 
unitary matrix $U$ such that $U_{S}=X$ for some set $S$ of the columns of $U$. The unitary $U$ can then be compiled into a sequence of Givens rotations which yields a 
Givens circuit $\mathcal{G(U)}$ corresponding to the unitary $U$. Applying the Givens circuit $\mathcal{G(U)}$ to the standard basis state $\ket{S}$ produces the subspace state 
$\ket{Col(X)}$. 

This template can be used together with any procedure for decomposing the unitary $U$ into a sequence of Givens rotations. Different decompositions of the matrix as product of Givens rotation matrices yield different subspace state preparation algorithms. The number of Givens rotations needed to express a unitary matrix $U$ can be viewed as a measure of the complexity for subspace state preparation. 
\begin{defn} 
Let $U \in \R^{n \times n}$ be a unitary matrix, the Givens complexity $\gamma(U)$ is defined as the minimum number of Givens rotation matrices in a decomposition $U = \prod_{(i, j, \theta) \in \mathcal{T} } G(i, j, \theta) $. 
\end{defn} 
\noindent A counting argument shows that the Givens complexity of a generic unitary matrix is $O(n^{2})$, however for well-structured matrices the Givens complexity can be much lower. The complexity of 
subspace state preparation for a matrix $X \in \R^{n \times d}$ using Givens circuits is given as, 
\al{ 
\gamma(X) = \min_{U} \{ \gamma(U) | \;\; \exists S \;\;s.t. \;\;U_{S} = X \} 
} 
The subspace state $\ket{Col(X)}$ can be created with cost $\gamma(X)$ if an efficient decomposition of the embedding unitary into a product of Givens rotations is known. 
In this direction, we note that several families of matrices including the Fourier, wavelet and several other well structured matrices \cite{rusu2021fast}  are known to have $\gamma(X)=O(n \log n).$
Matrices with low Givens complexity have also attracted attention recently in the machine learning literature \cite{frerix2019approximating}.

The worst case complexity of decomposing a unitary into a sequence of Givens rotations is the same as that for the Singular Value Decomposition algorithm, namely the pre-processing overhead to compute the Givens circuits for the $X \in \R^{n \times d}$ matrix is $O(nd^{2})$. Such pre-processing costs can be amortized in an online setting or in machine learning applications where the parameters are being updated frequently. We also provide subspace state preparation algorithms with linear pre-processing overhead of $O(nd)$  using logarithmic depth circuits called Clifford loaders in the next section.

In the next subsections, we provide worst case methods for decomposing a unitary into Givens rotations that are advantageous in terms of the circuit depth and total number of gates used respectively.

\subsection{Quantum Pyramid Circuits} 
It is desirable to have a Givens circuit that uses 
only $FBS_{ij} (\theta)$ gates with $|i-j|=1$, since such a circuit avoids the extra overheads incurred by longer range $FBS$ gates and can be implemented using $RBS$ gates only. 
In fact, a matrix with orthonormal columns can be compiled in time $O(nd^2)$ into a sequence of Givens rotations acting on adjacent rows, this construction corresponds to the quantum pyramid circuits constructed in a recent work on quantum orthonormal neural networks \cite{kerenidis2021classical}. 

\begin{theorem} 
A matrix $U \in \R^{n \times d}$ with orthonormal columns can be compiled into a Givens circuit using only $RBS$ gates in time $O(nd^2)$. The resulting pyramid circuit 
has depth $n+d$ and uses $(2n-1-d)d$ gates acting on adjacent qubits. 
\end{theorem} 
\noindent The number of $RBS$ gates used in the pyramid circuit is equal to the number of free parameters for a matrix $U \in \R^{n \times d}$ with orthonormal columns, so this construction is optimal in the number of gates. 
The circuit depth for the pyramid circuit is $O(n+d)$ and we will later see examples of constructions of Givens circuits that reduce the depth to $O(d \log n)$. The pyramid circuit can be used as a representation 
of the subspace, such an approach may be advantageous for hardware with limited qubit connectivity like grid-based architectures. 

\subsection{ Sine-cosine decomposition} 

An alternative decomposition of a unitary into Givens matrices can be found using the sine-cosine decomposition which is closely related to the singular value decomposition. The sine-cosine decomposition \cite{van1985computing}
recursively decomposes the orthogonal matrix $X$ as a product, 
\al{ 
X = \begin{pmatrix} U_{1} & 0 \\  0 & U_{2} \end{pmatrix} \begin{pmatrix} C & S \\  -S & C \end{pmatrix} \begin{pmatrix} V_{1} & 0 \\  0 & V_{2} \end{pmatrix}, 
} 
where $U_{1}, U_{2}, V_{1} , V_{2}$ are unitary matrices of size $n/2$ and $C$ and $S$ are diagonal matrices with entries  $\cos(\theta_{i})$ and $\sin(\theta_{i})$ on the diagonals. 
The sine cosine decomposition algorithm for partial isometries \cite{stewart1982computing, gawlik2018backward} can be used to obtain Givens circuits with $O(nd)$ gates, depth $O(d \log n)$ and pre-processing overhead $O(nd^{2})$.

\begin{algorithm} [H]
\caption{Subspace state preparation with Givens circuits.} \label{SSP1} 
\begin{algorithmic}[1]
\REQUIRE Orthogonal matrix $U \in \R^{n \times n}$ and subset $S \subset [n]$ of the rows or columns of $U$. 
\ENSURE The subspace state $\ket{Col(U_{S})}$ or $\ket{Row(U_{S})}$. 
\STATE Compute Givens decomposition for $U= \prod_{(i, j, \theta) \in \mathcal{T}} G(i, j, \theta)$ using either quantum pyramid circuits or the sine cosine decomposition. 
\STATE Compute Givens circuit $\mathcal{G}(U)$ corresponding to the decomposition computed in step 1. 
\STATE Output $\mathcal{G}(U) \ket{S} = \ket{Col(U_{S})}$ or $\mathcal{G}(U)^{\dagger} \ket{S} = \ket{Row(U_{S})}$. 
\end{algorithmic}
\end{algorithm}

\noindent The subspace state preparation algorithm with Givens circuits is described as algorithm \ref{SSP1}. The algorithm requires a
Givens decomposition of $U$ that can be computed using pyramid circuits or the sine-cosine decomposition and constructs the corresponding Givens circuit $\mathcal{G}(U)$. The subspace states corresponding to $\ket{Col(U_{S})}$ can then be generated by running the Givens
circuits with $\ket{S}$ as an input. Note that the Givens circuit corresponding to the sine-cosine decomposition is structured like a classical butterfly network with layers of $FBS_{i, i+2^{k}} (\theta)$ gates being applied in parallel. We leave open the optimal implementation of the above circuit with minimal overheads in terms of extra qubits and circuit depth, in the next section we provide such an optimal implementation of a different circuit consisting of $FBS$ gates.

\section{Subspace states and Clifford Loaders} \label{cliffload} 

The subspace state preparation algorithm using Givens circuits incurs an $O(nd^{2})$ pre-processing overhead which at times can be prohibitive. In this section, we present a different approach to subspace state 
preparation with a pre-processing overhead of $O(nd)$. The resulting quantum circuits that we term Clifford loaders are closely related to the correspondence between anti-commuting operator systems
and quantum systems on $n$ qubits that can be formalized using Clifford algebras. 

\subsection{Anti-commuting operator systems} 
This section introduces anti commuting operator systems and shows that they correspond to quantum systems with a certain number of qubits. This correspondence underlies the construction of the 
Clifford loaders circuits. 
\begin{defn} 
An anti-commuting operator system of rank $k$ is a set of operators $A_{i}, i \in [k]$ acting on a finite dimensional Hilbert space $H$ such that the relations, 
\al{ 
\frac{1}{2}(A_{i} A_{j} + A_{j} A_{i} )= \delta_{ij} I 
} 
are satisfied for all $i, j \in [k]$. 
\end{defn} 
\noindent Note that each operator in an anti-commuting system is a reflection operator as $A_{i}^{2} = I$. 
The algebra generated by an anti-commuting operator systens of a rank $k$ is isomorphic to a Clifford Algebra of rank $(k-1)/2$, 
implying that $dim(H) \geq 2^{(k-1)/2}$ \cite{brauer1935spinors}. 

We give a different proof of this classical result in the language of quantum information, by establishing a correspondence between rank $(2n+1)$ systems of mutually anti-commuting operators 
and $n$ qubit quantum systems. Jordan's lemma is the main technical tool used to establish this correspondence. 
\begin{lemma} \label{jordan} 
[Jordan's lemma, \cite{J75} ] Let $P$ and $Q$ be two projections acting on a finite dimenisonal Hilbert space $H$ and let $(v_{i}, \lambda_{i})$ be the eigenvectors and eigenvalues for $PQP$, then $H$ decomposes as a 
direct sum of at most two dimensional subspaces $(v_{i}, Qv_{i})$ such that $\braket{v_{i}} {Qv_{i}}^{2}  = \lambda_{i}$. 
\end{lemma} 
\noindent The product of the reflections $2P-I, 2Q-I$ also has the same invariant subspace decomposition as given by Jordan's lemma.

The correspondence between qubits and anti-commuting operators using Clifford algebras first appeared in the quantum information literature in Tsirelson's work on Bell inequalities \cite{tsirel1987quantum} where it yields a bound on the dimension of quantum systems required to achieve certain correlations. More recently, this correspondence was also used to construct explicit examples of $n$ dimensional matrices having sub-exponentially large CPSD rank \cite{prakash2018completely}. The following argument establishing 
the correspondence is implicit in the seminal work \cite{reichardt2013classical} on the verification of quantum computations. 

\begin{lemma} 
If there is an anti-commuting operator system of rank $(2n+1)$ acting on a Hilbert space $H$, then $H$ decomposes into a tensor product of $n$ qubit systems, in particular $dim(H) \geq 2^{n}$.  
\end{lemma} 
\begin{proof} 
We show that for every pair of anti-commuting operators $A_{i}, A_{j}$  in the anti-commuting system, there is a factorization of $H = H_{1} \otimes H'$ where $dim(H_{1})=2$ and $H'$ supports a system of mutually anti-commuting operators $B_{k}$ of rank $(2n-1)$ that commute with $A_{i}$ and $A_{j}$. The result follows by applying this construction iteratively. 

Let $A_{i}$ and $A_{j}$ be an arbitrary pair of operators in the anti-commuting system. Jordan's lemma applied to $A_{i}$ and $A_{j}$ yields a decomposition of $H$ into invariant spaces of dimension at most $2$. 
Further, the anti-commutation relation $A_{i} A_{j} + A_{j} A_{i} =0$ implies that all the invariant subspaces are 2 dimensional and that $A_{i}A_{j}$ acts as rotation by $\pm \pi/2$ on all invariant subspaces. 

Define operators $B_{k} = A_{i} A_{j} A_{k}$, then it follows that the $B_{k}$ for $k\neq i, j$ are $(2n-1)$ mutually anti commuting operators that commute with $A_{i}$ and $A_{j}$. 
The Jordan blocks for $A_{i}$ and $A_{j}$ induce a decomposition  $H= H_{1} \otimes H'$ as a tensor product of a qubit with Hilbert space $H'$ with $dim(H')= dim(H)/2$. The operators $A_{i}$ and $A_{j}$ 
can without loss of generality be assumed to be $X$ and $Z$ Pauli matrices acting on $H_{1}$. Operators commuting with $A_{i}$ and $A_{j}$ act as $I$ on $H_{1}$ and the $B_{k}$ 
can be decomposed as $B_{k} = I \otimes B_{k}'$ where $B_{k}'$ are an anti-commuting system of rank $(2n-1)$ supported on $H'$. Iterating the argument, it follows that $H= H_{1} \otimes H_{2} \otimes \cdots H_{n} \otimes H'$, 
where $H_{i}$ are $2$ dimensional Hilbert spaces and $H'$ is some Hilbert space.

\end{proof} 
\noindent Next, we see some canonical examples of anti-commuting operator systems. 
The Pauli matrices are a rank $3$ system of anti-commuting matrices acting on a Hilbert space of dimension $2$.  
\al{ 
I=\en{ \begin{matrix}  &1,&0 \\ &0,&1  \end{matrix} } , X=  \en{ \begin{matrix} &0,&1 \\ &1,&0  \end{matrix}} ,  Y= \en{\begin{matrix}  &0,&i \\ &-i,&0  \end{matrix} } , Z=\en{ \begin{matrix}  &1,&0 \\& 0,&-1  \end{matrix} } 
} 
These matrices satisfy the relations $XY=-YX=iZ, YZ=-ZY=iX$ and $XZ=-ZX=iY$. 

The Weyl-Brauer representation of the Clifford algebra for $\C^{n}$ is generated by the matrices 
$P_{i} =  Z^{\otimes i-1} \otimes  X \otimes I^{\otimes n-i} $ and $Q_{i} = Z^{\otimes i-1} \otimes Y \otimes I^{\otimes n-i}$ acting on $n$ qubits. It is easy to verify that the Weyl Brauer matrices together with $R=Z^{\otimes n}$ form an anti-commuting system of rank $(2n+1)$ supported on a Hilbert space of dimension $2^{n}$. 
The anti commuting systems considered in this work are generated by the operator $P_{i}$, that is the operators generate  the Clifford algebra for $\R^{n}$.

\subsection{Clifford Loaders} 
In this section, we introduce unitary operators corresponding to the representation of unit vectors in the Clifford algebra and show that these operators are efficiently implemented by a class of quantum 
circuits that we term Clifford loaders.  

\begin{defn} \label{Cliff} 
For unit vector $x \in \R^{n}$ with $\norm{x}_{2} =1$, the operator $\Gamma(x)$ is defined as, 
\all{ 
 \Gamma(x)=  \sum_{i \in [n]} x_{i} Z^{\otimes (i-1)} \otimes X \otimes I^{\otimes (n-i)}. 
 } {eq3} 
\end{defn} 
\noindent Note that $\Gamma(x)= \sum_{i} x_{i} P_{i}$ for anti-commuting operators $P_{i}$, the anti-commutation relations imply that $\Gamma(x)^{2} = I$ for unit vectors $x \in \R^{n}$,
\al{ 
\Gamma(x)^{2} = \sum_{i,j} x_{i} x_{j} (P_{i} P_{j} + P_{j} P_{i}) + \sum_{i} x_{i}^{2} P_{i}^{2} =  I 
} 
It follows that for unit vectors $x \in \R^{n}$, the operator $\Gamma(x)$ is unitary, in fact it is a reflection as its eigenvalues are constrained to be $\pm 1$. 
The action of $\Gamma(x)$ on standard basis states is given by, 
\all{ 
\Gamma(x) \ket{S} = \sum_{i \in [n]} (-1)^{ \oplus_{j<i} s_{j}} x_{i} \ket{S \oplus i } . 
} {cliff} 
The operators $\Gamma(x)(x)$ are unitary matrices of dimension $2^{n}$ acting on an $n$-qubit Hilbert space,  and there is no a priori reason to expect that they can be implemented efficiently. 
The main result of this section is that these operators can indeed be implemented efficiently with polynomial resources, using quantum circuits that we term Clifford loaders. 

A unary data loader is a parametrized quantum circuit using RBS gates for preparing quantum states corresponding to the unary amplitude encodings of vectors defined in \ref{def1}. 
\begin{defn} \label{dataloader} 
A unary data loader $D(x)$ is a parametrized quantum circuit consisting of RBS gates that for any unit vector  $x \in \R^{n}$ maps the standard basis state $\ket{e_1 }$ to the unary amplitude encoding $\ket{un(x)}$, 
\al{ 
D(x) \ket{e_1 } =  \sum_{i \in [n]} x_{i}  \ket{ e_i }, 
}
where $e_i=0^{i-1}10^{n-i}.$
\end{defn} 
\noindent Optimal constructions of data loaders with $O(n)$ gates and depth $\log n$ have been given \ref{}. Note that the circuit $D'(x)$ obtained by replacing each $RBS$ gate in a unary data loader $D(x)$ 
by the corresponding $FBS$ gate is also a unary data loader as the action of RBS and FBS gates is the same on bit strings with Hamming weight $1$. Clifford loader circuits are obtained by composing such fermionic data loaders, where each $RBS$ gate in the data loader has been replaced with the corresponding $FBS$ gate. 

\begin{defn} \label{cloader} 
Given a unary data loader $D(x)$, the Clifford loader is the circuit $\mathcal{C}(x) = D_F(x)(X\otimes I^{\otimes n-1}) D_F(x)^{*}$ where $D_F(x)$ is obtained by replacing each $RBS$ gate in $D(x)$ by the corresponding $FBS$ gate. 
\end{defn} 
\noindent Unlike a unary data loader $D(x)$ whose behavior is specified only for the input state $\ket{e_1}$ and can be arbitrary on other standard basis states, 
a Clifford loader is determined on the entire Hilbert space. The following result describes the action of the Clifford loader on subspace states of arbitrary dimension, which implies that the Clifford loaders in definition 
\ref{cloader} do not depend on the choice of the unary data loader $D(x)$. 
\begin{theorem} \label{action:cloader} 
Let $x = \cos(\theta) x^{\parallel} + \sin(\theta) x^{\perp}$ be the decomposition of unit vector $x \in \R^{n}$ into components orthogonal and parallel to a $d$-dimensional subspace $Col(Y)$. Then, 
\al{ 
\mathcal{C} (x) \ket{Col(Y)} = \cos(\theta) \ket{Col(Y' )} + \sin(\theta)  \ket{Col( (Y , x^{\perp}) )}
} 
where the columns of $Y'\in \R^{n\times (d-1)}$ together with $x^{\parallel}$ are an orthonormal basis for $Col(Y)$.  
\end{theorem} 
\begin{proof} 
The unitary invariance of subspace states implies that $\ket{Col(Y)} = \ket{Col( (x^{\parallel}, Y') )}$. 
Let $x'$ be the orthogonal complement of $x$ in the two dimensional subspace spanned by $(x^{\parallel}, x^{\perp})$ so that $x^{\parallel} = \cos(\theta) x + \sin(\theta)x'$. Let $Y_{1}= (x, Y')$ and $Y_{2} = (x', Y')$ be the 
matrices obtained by adjoining vectors $x, x'$ to the columns of $Y'$.

Let the Clifford loader be $ D_F(x)(X\otimes I^{\otimes (n-1)}) D_F(x)^{*}$, note that $D_F(x)$ is a Givens circuit for some unitary $U_{D} \in \R^{n\times n}$ such that $U_{D} e_{1} = x$. The result of applying $D_F(x)^{*}$ on $\ket{ Col(Y)}$ can be computed using 
Theorem \ref{givens}. 
\al{ 
D_F(x)^{*} \ket{Col(x^{\parallel}, Y')} =  \ket{ Col( \cos(\theta) U_{D}^{*} Y_{1}  + \sin(\theta) U_{D}^{*} Y_{2} )  } 
} 
From the definition of $U_{D}$ it follows that the first row of the matrix $V_{1} = U_{D}^{*} Y_{1}$ is the vector $(1, 0^{d-1})$  and that the first row of $V_{2} = U_{D}^{*} Y_{2}$ is the all zeros vector 
$0^{d}$. The remaining part of the Clifford loader circuit $D_F(x)(X\otimes I^{\otimes (n-1)})$ can therefore be computed as, 
\al{ 
D_F(x)(X\otimes I^{\otimes (n-1)}) \ket{ Col( \cos(\theta) V_{1} + \sin(\theta) V_{2} )  } = \cos(\theta)  \ket{ Col( Y') }  + \sin(\theta) \ket{ Col(Y, x^\perp) }   
} 
The first part of the superposition has Hamming weight $(d-1)$ and represents the subspace obtained by removing $x^{\parallel}$ from $Y$ while the second part of the superposition has Hamming weight $(d+1)$ and represents the subspace obtained by adding 
$x^{\perp}$ to $Y$.  
\end{proof} 

\noindent We next show that the Clifford loader implements the unitary $\Gamma(x)$ defined in equation \eqref{eq3}. In order to prove this, we compute the unitary implemented by a Clifford loader explicitly using the a simple construction of a unary data loader based on spherical coordinates and the following identities on conjugation by the $RBS$ gate.

\begin{lemma} \label{bsconj} 
The following identities hold: 
\enum{ 
\item $RBS(\theta) (X\otimes I) RBS(\theta)^{*} =  \cos(\theta) X\otimes I + \sin(\theta) Z \otimes X$. 
\item  $RBS(\theta) (Z\otimes Z) RBS(\theta)^{*} = Z\otimes Z$. 
} 
\end{lemma} 
\begin{proof} 

The matrices $X\otimes I$ and $Z \otimes X$ have the following representation in the standard basis, 
\al{ 
X\otimes I = \begin{pmatrix} 
&0  &0 & 1 & 0 \\
&0  & 0 &0  &1 \\ 
&1 & 0   & 0    & 0 \\
&0 & 1 & 0     & 0 
\end{pmatrix} , \hspace{50pt} 
Z \otimes X = \begin{pmatrix}  
&0  &1 & 0 & 0 \\ 
&1  & 0 &0  &0 \\
&0 & 0   & 0    & -1 \\
&0 & 0 & -1     & 0 
\end{pmatrix}. 
}  
By direct computation we have the first identity, 
\all{ 
RBS(\theta) (X \otimes I) RBS(\theta)^{*} =  \begin{pmatrix} 
&0  &\sin(\theta) & \cos(\theta) & 0 \\
&\sin(\theta)  &0 & 0  &\cos(\theta) \\
&\cos(\theta) & 0   & 0  & -\sin(\theta) \\
&0 & \cos(\theta) & -\sin(\theta) & 0 
\end{pmatrix}  = \cos(\theta) X \otimes I + \sin(\theta) Z \otimes X. 
} {bconj} 
The second identity also follows by direct computation, 
\all{ 
RBS(\theta) (Z \otimes Z) RBS(\theta)^{*} =  \begin{pmatrix} 
&1  &0 & 0 & 0 \\
&0  & -1 &0  &1 \\ 
&0 & 0   & -1    & 0 \\
&0 & 0 & 0     & 1 
\end{pmatrix} = Z \otimes Z. 
}{i2} 
\end{proof} 
\noindent As noted before the $RBS$ gate can be viewed as a special case of the $FBS$ gate with $|i-j|=1$. 
The conjugation identities can be generalized to conjugation by $FBS$ gates for arbitrary $i, j$. However, conjugation 
identities for the $RBS$ gates suffice to establish that $\Gamma(x)$ is the unitary implemented by the Clifford loader circuits. 

\begin{figure}
	\centering
	\includegraphics[scale = 0.8]{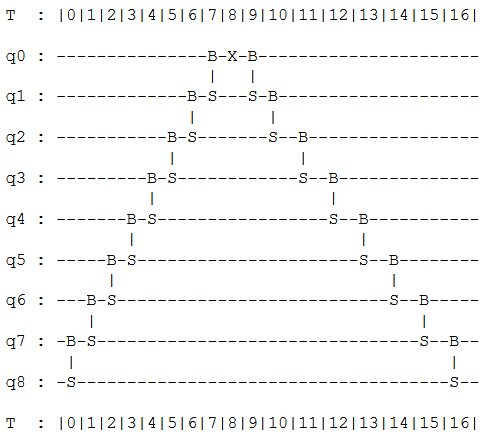}
	\caption{Quantum circuit corresponding to the first construction of the Clifford loader $\mathcal{C}(x)$ on 8 qubits. } \label{fig1} 
\end{figure}

\begin{theorem} \label{cliff1} 
For all unit vectors $x \in \R^{n}$ with $\norm{x}_{2}=1$, the Clifford loader circuit $\mathcal{C}(x)$ implements the unitary $\Gamma(x)$. 
\end{theorem} 
\begin{proof}
Given a unit vector $x \in \R^{n}$ it is possible to compute in time $O(n)$, a sequence of angles $\theta_{i}$ for $i\in [n-1]$ such that, 
\all{ 
x_{1} &= \cos(\theta_{1}) \nl
x_{2} &= \cos(\theta_{2}) \sin(\theta_{1})  \nl
x_{3} &= \cos(\theta_{3})  \sin(\theta_{1}) \sin(\theta_{2}) \nl
& \cdots \nl 
x_{n-1} &= \cos(\theta_{n-1}) \prod_{1 \leq i < (n-1)} \sin(\theta_{i}) \nl 
x_{n} &= \sin(\theta_{n-1}) \prod_{1 \leq i < (n-1)} \sin(\theta_{i}). 
} {angles} 
The angle sequence $(\theta_{1}, \theta_{2}, \cdots , \theta_{n-1})$ for $x$ is in correspondence with the spherical coordinate representation for $x$. 
A unary data loader $D(x):= RBS_{n, n-1}(\theta_{n-1})  RBS_{n-1, n-2}(\theta_{n-2}) \cdots RBS_{1,2} (\theta_{1})$ can be constructed as a sequence of $RBS$ gates with angles given as above. 
This data loader circuit $D(x)$ is a sequence of $(n-1)$ $RBS$ gates applied sequentially to consecutive qubits. 

As these $RBS$ gates are being applied to consecutive qubits, they are identical to the $FBS$ gates and $D(x)=D_F(x)$. The quantum circuit $D(x) (X \otimes I^{\otimes n-1}) D(x)^{*}$ is therefore a Clifford loader, and this circuit is illustrated in Figure \ref{fig1}. 
It remains to show that this Clifford loader circuit correctly implements the unitary $\Gamma(x)$.

Lemma \ref{bsconj} is used to establish the correctness of the construction by iteratively computing the results of the conjugations by $RBS$ gates, 
\al{ 
RBS_{1,2} (\theta_{1}) (X \otimes I^{\otimes (n-1)}) (RBS_{1,2} (\theta_{1}) )^{*} &= \cos(\theta_{1}) X \otimes I^{\otimes (n-1)} + \sin(\theta_{1}) Z\otimes X \otimes I^{\otimes (n-2)} \nl 
&= x_{1}  X \otimes I^{\otimes (n-1)} + \sin(\theta_{1})  Z \otimes X \otimes I^{\otimes (n-2)}. 
} 
Iterating this procedure and conjugating by $RBS_{2,3}(\theta_{2})$, 
\al{ 
RBS_{2,3} (\theta_{2}) (x_{1}  X \otimes I^{\otimes (n-1)} + \sin(\theta_{1})  Z \otimes X \otimes I^{\otimes (n-2)}) (RBS_{2,3} (\theta_{2}) )^{*} = &x_{1} X \otimes I^{\otimes (n-1)}  + x_{2} Z \otimes X \otimes I^{\otimes (n-2)} +  \nl 
&\sin(\theta_{1}) \sin(\theta_{2}) Z\otimes Z\otimes X \otimes I^{\otimes (n-3)}. 
} 
Continuing iteratively and using equation \eqref{angles} for the final step, it follows that after the $(n-1)$ conjugations by $RBS_{j, j+1} (\theta_{j})$ the operator $\Gamma(x)$ has been implemented.  

\end{proof} 

\noindent The above theorem provides an explicit construction for a linear depth Clifford loader, this circuit is illustrated in Figure \ref{fig1}. The angles for the $RBS$ gates are computed as in equation \eqref{angles}. 
A logarithmic depth Clifford loader circuit can be obtained using logarithmic depth unary data loaders, this construction is detailed in the following section.

\subsection{Logarithmic depth Clifford loaders} 
The logarithmic depth Clifford loader is obtained from the logarithmic depth construction of unary data loaders (see Figure \ref{fig3}). The $FBS$ gates can be implemented as described in proposition \ref{fbsimp}. 

However, the number of gates in the logarithmic depth Clifford loader can be further optimized by amortizing the cost of computing the parities. Let $D(x)$ be the logarithmic depth $n$-dimensional unary data loader with $FBS$ gates and let $\tilde{D}(x)$ be the circuit $D(x)$ composed with a sequence of $CX$ gates so that qubit 2 contains the parity  of all qubits from 2 through $n$ at the end of the computation. The Clifford loader circuit $\mathcal{C}(x)$ uses one half of the $CX$ gates 
in the implementation of the $FBS$ gate in Claim \ref{fbsimp}. Note that for a two-dimensional unit vector, we have $D(x)= \tilde{D}(x) = RBS(\theta)$.

The circuits $D(x)$ and $\tilde{D}(x)$ are then specified by the following recursive definitions,  
\al{ 
D(x) &= ( \tilde{D}(x^{1}) || D(x^2), FBS_{1, (n/2)+1} ( \theta_0 ) ) \nl 
\tilde{D}(x) &=  ( \tilde{D}(x_{1}) || \tilde{D}(x_{2}), FBS_{1, (n/2)+1} (\theta_0 ) , CX_{(n/2+2, n/2+1)} , CX_{( n/2+1, 2)}  ) 
} 
The unitary $D(x)$ implements the Clifford loader with $\mathcal{C}(x) = D(x) (X \otimes I^{\otimes (n-1)}) D(x)^{*})$, note that this construction optimizes the logarithmic depth Clifford loader circuit eliminating pairs of redundant $FBS$ gates. 

The circuit depth for $D(x)$ can be obtained from the recursive relations, let $d(n)$ and $d'(n)$ be the circuit depths for $D(x)$ and $\tilde{D}(x)$ as a function of dimension. We have $d(2)=d'(2)=1$ and from the recursion we have $d(n)= d'(n/2) + 3$ and $d'(n) = d'(n/2)+4$. Note that the gadget G has depth 3 and the $CX_{(n/2+2, n/2+1})$ can be performed in parallel with the third layer of $FBS_{1,(n/2)+1} ( \theta_0 )$, when we implement the circuit using these recursive relations. The explicit solution for these recurrences is $d(n) = 4 (\log n -1)$ where $n>2$ is a power of 2. The quantum circuits corresponding to this construction are illustrated in Figure \ref{fig2}.

\begin{figure}[H] \label{fig2} 
	\centering
	\includegraphics[scale = 0.8]{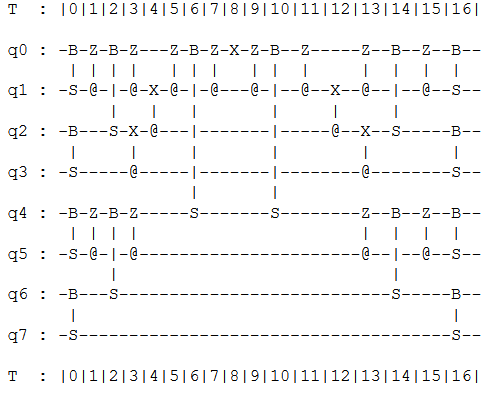}
	\caption{Quantum circuit corresponding to logarithmic depth construction of the Clifford loader for 8 qubits. The circuits to the left and right are the unary data loader in Figure \ref{fig3} with $RBS$ gates replaced by implementations of the corresponding $FBS$ gates. }
\end{figure}

\section{Quantum Determinant Sampling} 
We will now use the tools introduced in the previous sections to obtain three new applications in quantum machine learning. We start with
Determinant Sampling which is closely related to the topic of Determinantal Point Processes (DPPs) in machine learning \cite{derezinski2021determinantal}. Determinant sampling is an importance sampling technique with several applications to classical linear algebra and machine learning including least squares problems and 
low rank approximation.  The determinant sampling problem is defined for a full rank matrix $A \in \R^{n\times d}$ as follows. 
\begin{defn} 
Given a full rank matrix $A \in \R^{n \times d}$  the determinant sampling problem for $A$ is the problem of sampling from the probability distribution on $S \subset [n], |S| =d$ such that $p(S) = \frac{ det(A_{S})^{2} } { det(A^{T}A) }$ where $A_{S}$ denotes the sub-matrix of $A$ with rows indexed by $S$. 
\end{defn} 
\noindent  Note that the full rank condition is equivalent to $det(A^{T}A) > 0$ and required for the determinant distribution to be well-defined.  The Cauchy-Binet identity stated as Theorem \ref{cbi} shows that this is indeed a probability distribution as the probabilities sum up to $1$. 

For all  $S \subset [n], |S| =d$, the squared determinants $ det(A_{S})^{2} $ are invariant under sign flips and column operations applied to $A$. The determinant distribution therefore depends only on the underlying subspace $\mathcal{A}= Col(A)$ and not on the representing matrix $A$. It follows from the definition \ref{def2} of subspace states, that measuring the subspace state $\ket{Col(A)}$ in the standard basis produces exactly a sample from the determinant distribution for $A$. 

Classical determinant sampling algorithms also utilize the invariance of the determinant distribution under column operations. The first step for the classical algorithms is to orthogonalize the matrix $A$, that is to find an orthonormal matrix $\tilde{A}$ with the same column space as $A$. The algorithms then perform determinant sampling from the orthogonal matrix $\tilde{A}$. The best known classical determinant sampling algorithm with this approach requires $O(nd^{2})$ time to pre-process the matrix to have orthogonal columns and requires $O(d^{3})$ arithmetic operations to generate each sample from the determinant distribution \cite{derezinski2019minimax}. This algorithm was discovered recently and is an improvement on previous approaches that required $O(nd^{2})$ pre-processing and $O(d^{4})$ arithmetic operations for generating samples. There are also alternative classical algorithms for determinant sampling using Markov chain Monte Carlo methods that incur larger polynomial overheads. 

The pre-processing for the quantum determinant algorithms is the same as that for the classical algorithm, namely the matrix $A$ is orthogonalized to $\tilde{A}$ such that $Col(A)= Col(\tilde{A})$ in time $O(nd^{2})$. 

Note that the subspace state preparation algorithm \ref{SSP1} also provides an algorithm for quantum determinant sampling given an embedding of $A$ into the columns of a unitary matrix $U$. 
In this case as well, there is a pre-processing overhead of $O(nd^{2})$ for finding the embedding by orthogonalization of $A$ followed by computation of the Givens decomposition for $U$. 
For $d=O(n)$, the worst case Givens complexity is $O(n^{2})$, so the quantum determinant sampling algorithm has complexity $O(n^{2})$ compared to the classical $O(n^{3})$ algorithm. 
If an embedding of $A$ into a unitary $U$ having Givens complexity $O(n \log n)$ is known, then there is a potentially cubic speedup  as the quantum algorithm requires $O(n \log n)$ operations for generating a sample from the determinant distribution compared to the best known classical algorithm that still requires $O(n^{3})$ operations.

In this section we present another quantum determinant sampling algorithm (Algorithm \ref{QDS}) using the  Clifford loaders defined in section \ref{cliffload} for subspace state preparation. 
The quantum determinant sampling algorithm  \ref{QDS} is stated for matrices with orthonormal columns, the Clifford loaders used in the algorithm can be constructed in a single pass over the 
matrix in time $O(nd)$. 

\begin{algorithm} [H]
\caption{Quantum Determinant Sampling} \label{QDS} 
\begin{algorithmic}[1]
\REQUIRE Matrix $A \in \R^{n \times d}$ such that $A^{T} A = I_{d}$, let $a^{1}, a^{2}, \cdots, a^{n}$ be the columns of $A$. 
\ENSURE A sample from the determinant distribution $\Pr[S] = \frac{ det(A_{S})^{2} }  { det(A^{T}A) }$. 
\STATE Starting with state $\ket{0^{n}}$ apply the Clifford loaders corresponding to the $d$ columns of $A$, to get
\al{ 
\ket{\mathcal{A}}  = \mathcal{C}(a^{d})  \cdots \mathcal{C}(a^{2}) \mathcal{C}(a^{1}) \ket{0^{n} } 
} 
\STATE Measure $\ket{\mathcal{A}} $ in the standard basis to obtain bit string corresponding to $S \subset [n]$. Output $S$. 
\end{algorithmic}
\end{algorithm}

\noindent It follows from the logarithmic depth Clifford loader construction that the quantum circuit used in algorithm \ref{QDS} has depth $O(d \log n)$ and uses $O(nd \log n)$ gates. The correctness follows from Lemma \ref{action:cloader} 
which shows that $\mathcal{C}(a^{d})  \cdots \mathcal{C}(a^{2}) \mathcal{C}(a^{1}) \ket{0^{n} }$ is the subspace state corresponding to the column space of $A$
for matrices $A$ with orthonormal columns. 

We provide a different proof for the correctness of Algorithm \ref{QDS} in this section, this proof enables the efficient computation of the amplitudes for sequential compositions of the Clifford loaders. 
For sets $S, T \subset [n]$ let $P_{k}(T, S)$ denote the set of length $k$ paths from $T$ to $S$ on the hypercube $\mathcal{H}$. A path $p \in P_{k}(T, S)$ is given as a sequence $(p(1), p(2), \cdots, p(k)) \in [n]^{k}$, where the value $p(k)$ represents the index of the bit flipped at the $k$-th step of the path. 
\begin{lemma} \label{signs} 
Let $x_{1}, x_{2}, \cdots, x_{k} \in \R^{n}$, then the coefficient of $e_{T \oplus S}$ in $(\prod_{i \in [k]} \mathcal{C}(x_{i}) ) \ket{e_{T}}$ is, 
\al{ 
\sum_{p \in P_{k}(T, S)} (-1)^{Sgn(p, T)} (x_{1})_{p(1)}, (x_{2})_{p(2)}, \cdots (x_{k})_{p(k)} 
} 
where $Sgn(p,T)$ is the parity of the number of swaps made to sort $(p, T)$ in ascending order where $T$ is already sorted in ascending order. 
\end{lemma} 
\begin{proof} 
Expanding the terms using the fact that $\mathcal{C}(x_{i})= \Gamma(x_{i})$, it follows that $(\prod_{i \in [k]} \mathcal{C}(x_{i}) ) \ket{e_{T}}$ is a linear combination of terms of the form $\pm (x_{1})_{p(1)}, (x_{2})_{p(2)}, \cdots (x_{k})_{p(k)}  \ket{T \oplus \chi(p)} $ where $\chi(p) \subset [n]$ is the 
characteristic vector of the length $k$ path $p$. 

It follows that if $\ket{T \oplus \chi(p) } = \ket{ T \oplus S }$ then $p \in P_{k}(T,S)$. In order to determine the sign, observe that $x_{p(k), k}$ incurs a sign factor 
of $(-1)^{ T \cap [p(k)-1]}$ which is the same as the parity of swaps needed to sort $T': =(p(k), T)$ in ascending order given that $T$ is already sorted. The argument can be continued iteratively, this is analogous to 
sorting the list $(p, T)$ where $T$ is a sorted list using the Bubble sort algorithm. The overall sign factor 
is $Sgn(p,T)$, the parity of the number of swaps needed to sort $(p, T)$ in ascending order. 
\end{proof} 
\noindent Let $\mathcal{S}_{d}$ denote the symmetric group on $d$ symbols and $e_{\O}=0^n$ represent the empty set. A useful corollary of the above Lemma is the following: 
\begin{corollary}  \label{det} 
Let $S \subset [n], |S|=d$, then the coefficient of $e_{S}$ in $(\prod_{i \in [d]} \mathcal{C}(x_{i}) ) \ket{e_{\O} }$ is given by $\sum_{\sigma \in \mathcal{S}_{d}} (-1)^{Sgn(\sigma)} (x_{1})_{\sigma(s_1)} (x_{2})_{\sigma(s_2)} \cdots 
(x_{d})_{\sigma(s_d)}$.  
\end{corollary} 
\begin{proof} 
There is exactly one length $d$ path from $\O \to S$ up to permutation. The result now follows directly from Lemma \ref{signs}. 
\end{proof}

\noindent We are now ready to demonstrate the main result for this section, a quantum algorithm for determinant sampling for matrices with orthonormal rows,   
\begin{lemma} \label{signs} 
Let $A \in \R^{n \times d}$ be such that $A^{T} A = I_{d}$, then, 
\all{ 
\en{ \prod_{i \in [d]} \mathcal{C}(a^{i})  }  \ket{e_{\O}} = \sum_{|S|=d} det(A_{S}) \ket{e_{S}}  
} {six} 
\end{lemma} 
\begin{proof} 
Corollary \ref{det} implies that for $|S|=d$ the coefficient of $\ket{e_{S}} $ in $\en{ \prod_{i \in [d]} \mathcal{C}(a^{i})  } \ket{e_{\phi}} $  is given by $det(A_{S}) = \sum_{\sigma \in \mathcal{S}_{d}} (-1)^{Sgn(\sigma)}
a_{1, \sigma(s_1)} a_{2, \sigma(s_2)} \cdots a_{d, \sigma(s_d)}$ using the Laplace expansion of the determinant. Further, by the Cauchy-Binet identity 
$\sum_{|S|=d} det(A_{S})^{2} = det(A^{T} A ) = 1$ as the matrix has orthonormal rows. It follows that equation \eqref{six} is an exact equality. 
\end{proof} 

\noindent This establishes the correctness of the quantum determinantal sampling algorithm \ref{QDS}. Each of the $d$ circuits 
can be implemented with $O(n \log n)$ gates and depth $O(\log n)$, yielding the following result.

\begin{theorem} (Quantum Determinal Sampling)
Given $A \in \R^{n \times d}$, there is a quantum determinant sampling algorithm that requires $O(nd^{2})$ pre-processing to orthogonalize $A$, and generates subsequent samples 
using a quantum circuit with $O(nd \log n)$ gates and with total depth $O(d \log n)$. 
\end{theorem}

Note that given as input the matrix $A \in \R^{n \times d}$ with orthonormal columns, the quantum determinant sampling algorithm \ref{QDS} uses $\tilde{O}(nd)$ operations for pre-processing (reading the input matrix and computing the parameters of the Clifford loaders), while the algorithm based on the Givens decomposition needs to find the Givens decomposition decomposition of the input matrix $A$ which takes $O(nd^2)$ operations.

\section{Singular value estimation for compound matrices} 
The efficient construction of Givens circuits $\mathcal{G}(U)$ corresponding to a unitary $U$ can be used to obtain quantum linear algebra 
algorithms that operate on subspaces. The notion of a compound matrix \cite{horn2012matrix} is key to obtaining this new family of quantum algorithms. 
\begin{defn} \label{def:cmpd} 
[Compound matrix] Given a matrix $A \in \R^{n \times n}$, the compound matrix $\mathcal{A}^{k}$ for $k \in [n]$ is the $\binom{n}{k}$ dimensional matrix with entries 
$\mathcal{A}^{k}_{IJ} = det(A_{IJ})$ where $I$ and $J$ are subsets of rows and columns of $A$ with size $k$. 
\end{defn} 

\noindent Compound matrices preserve the property of being unitary (orthogonal), that is the compound matrices $\mathcal{U}^{k}$ for a unitary (orthogonal) matrix $U$ are also unitary. 
Compound matrices are multiplicative, this property is a consequence of the Cauchy Binet identity, 
\begin{claim} 
For matrices $A, B \in \R^{n \times n}$ and $0\leq d\leq n$, the multiplicative property $ (\mathcal{AB})^{k}= \mathcal{A}^{k} \mathcal{B}^{k} $ holds. 
\end{claim} 
\noindent The relation between the Givens circuits $\mathcal{G}(U)$ and the compound matrices for a unitary $U$ is given by the next Lemma which shows that 
compound matrices embed into the unitaries implemented by the Givens circuits. 

\begin{lemma} \label{lem:givens} 
\noindent The Givens circuit $\mathcal{G}(U)$ implements the unitary $\oplus_{k =0}^{n} \mathcal{U}^{k}$, i.e. the Givens circuit is the direct sum of the compound matrices. 
\end{lemma} 
\begin{proof} 
The Givens circuit $\mathcal{G}(U)$ is Hamming weight preserving, it is therefore sufficient to prove that $\mathcal{G}(U)= \mathcal{U}^{k}$ when restricted to the space $\mathcal{H}_{k}$
for $0\leq k\leq n$. This holds for $k=0$ and $k=n$ as these one dimensional subspaces are invariant under both $\mathcal{G}(U)$ and $ \mathcal{U}^{k}$.
The definition of the compound matrix  \ref{def:cmpd} implies that $\mathcal{U}^{k} \ket{S} = \ket{Col(U_{S})}$ for $1\leq k< n$, this is identical to the action of the Givens circuit $\mathcal{G}(U)$ by Theorem 
\ref{thm: action}.  
\end{proof} 

\noindent The eigenvectors and eigenvalues of a compound matrix $\mathcal{A}^{k}$ are determined completely by the spectrum of $A$. If the eigenvalues and eigenvectors of $A$ are $(v_{i}, \lambda_{i})$ then
the spectrum of $\mathcal{A}^{k}$ is indexed by subsets $S \subset [n], |S|=k$, the eigevectors are the subspace states $\ket{ Col(V_{S})}$ and the corresponding eigenvectors are $\prod_{i \in S} \lambda_{i}$.  
That is the eigenvectors for $\mathcal{A}^{k}$ are the subspace states corresponding to a group of $k$ eigenvectors for $A$ and the corresponding eigenvalue is the product of the corresponding eigenvalues for $A$. 

The above Lemma \ref{lem:givens} thus provides an explicit description of the spectrum for the Givens circuits in terms of the eigenvectors and eigenvalues of $U$. 

\begin{lemma} 
Let $(v_{i}, e^{i\theta_{i}})$ be the eigenvectors and eigenvalues of $U$, the eigenvectors of the Givens circuit $\mathcal{G}(U)= \oplus_{i =0}^{n} U^{k}$ are 
$\ket{Col(V_{S})}$ and the corresponding eigenvalues are $e^{i \sum_{i\in S} \theta_{i}}$. 
\end{lemma} 
\begin{proof}
The action of $\mathcal{G}(U)$ on $\ket{Col(V_{S})}$ can be computed using Lemma \ref{givens}, 
\al{ 
\mathcal{G}(U) \ket{Col(V_{S})} = \ket{ Col( UV_{S}) } = e^{i \sum_{i\in S} \theta_{i}} \ket{ Col(V_{S}) }. 
} 
Note that in general the eigenvectors $v_{i}$ are complex vectors, Lemma \ref{givens} continues to hold for complex vectors. 
\end{proof} 
\noindent This result yields a phase estimation algorithm for subspace states instead of vectors. Given a Givens circuit $\mathcal{G}(U)$ for a unitary $U$ and an arbitrary subspace state $\ket{Col(X)}$, the 
quantum phase estimation algorithm for $\mathcal{G}(U)$ decomposes it into a superposition over eigenstates of $\mathcal{U}^{k}$, that is as a linear combination of the subspace states $\ket{Col(V_{S})}$. 

Classically, an algorithm for carrying out such a decomposition would need to work with exponentially large compound matrices which do not have a good low-rank approximation, so it seems plausible that the phase estimation algorithm for subspace states can offer 
high degree polynomial or even exponential quantum speedups in some settings. In order to establish such an exponential speedup, it needs to be seen if the subspace state phase estimation algorithm can be dequantized, that is if there is an polynomial time classical algorithm to sample from the output of the phase estimation algorithm for subspace states. 

There are two reasons for expecting phase estimation for compound matrices to be more resistant to dequantization than low rank quantum linear algebra methods. First $\ell_{2}$-sampling from the rows or columns of compound matrices is equivalent to the determinant sampling problem, which has cubic cost in the classical case, increasing polynomial overheads for possible 'quantum inspired' algorithms. Second, even if the matrix $A$ has a rank $r$ approximation for $r \ll n$, the compound matrix $\mathcal{A}^{k}$ will have a low rank approximation with rank $r^{k}$ which is exponential in $r$ for large enough $k$. 

In addition, the singular value estimation and transformation algorithms  \cite{KP16, gilyen2020improved} also generalize to subspace states as we show next. First we generalize the notion of a block encoding, principal vectors and principal angles 
which are used to derive the singular value estimation and transformation algorithms to the setting of compound matrices. 

\begin{claim} \label{d5} 
(Generalized block encoding) Let $P, Q \in \R^{n\times n}$ be orthogonal matrices and let $A=(P^{T} Q)_{IJ}$ for $I, J \subseteq [n]$ be a block encoding of $A$ as a submatrix of $P^{T}Q$. Then, there is a block encoding for the compound matrix given by $\mathcal{A}^{k}= (\mathcal{P^{T}Q} )^{k}_{I_{k}, J_{k}} $ for all $k \leq \min(|I|, |J|)$ with index sets $I_{k}, J_{k}$ consisting of $k$-tuples from elements of $I$ and $J$. 
\end{claim}  
\noindent The matrix $P$ can be taken to be the identity matrix to recover the usual notion of block encoding. The above definition extends the notion of a block encoding to compound matrices, showing that 
an embedding of a matrix $A$ into unitary $U$ also implies an embedding for the corresponding compound matrices. The next result generalizes the relation between the singular value decomposition of $A$ and the principal angles between subspaces 
$Col(P^{T}_{I})$ and $Col(Q_{J})$ to the case of compound matrices.  

\begin{theorem} 
(Compound matrix SVD) Let $A=(P^{T} Q)_{IJ}$ be as in Definition \ref{d5} and let $A= U^{T} \Sigma V$ be the singular value decomposition of $A$. 
Then the principal vectors and angles between $\mathcal{P}^{k}_{I_{k}}$ and $\mathcal{Q}^{k}_{ J_{k}}$ are given by $( \ket{Col(PU_{S})},  \ket{Col(QV_{S})} )$ and $\cos(\theta_{S}) = \prod_{i \in S} \sigma_{i}$
for  $S \in \mathcal{H}_{k}$. 
\end{theorem} 
\begin{proof} 
Let us first recall the proof for $d=1$ using Jordan's lemma \ref{jordan} for the projectors $\Pi_{P}= P_{I}P_{I}^{T}$ and $\Pi_{Q}= Q_{J}Q_{J}^{T}$. The principal vectors pairs $(Pu_{i}, Qv_{i})$ for $i\in [n]$ are invariant under the action of $\Pi_{P}$
and $\Pi_{Q}$. The principal angles $\theta_{i}$ are the angles between the principal vectors, 
\al{ 
\cos(\theta_{i}) = \braket{Pu_{i} } {Qv_{i}} = \braket{u_{i} } {Av_{i}} = \sigma_{i}.
} 

For the general case define the projectors $\Pi_{P}^{k} = \mathcal{P}_{I_{k}}\mathcal{P}_{I_{k}}^{T}$ and $\Pi_{Q}^{k} = \mathcal{Q}_{I_{k}}\mathcal{Q}_{I_{k}}^{T}$. 
In order to prove that $( \ket{Col(PU_{S})},  \ket{Col(QV_{S})} )$ are a pair of principal vectors, it suffices to show that the subspace is invariant under the action of 
$\Pi_{P}^{k}$ and  $\Pi_{Q}^{k}$.  The multiplicativity of compound matrices yields the relations,  
\al{ 
\Pi_{P}^{k} \ket{Col(QV_{S})}  = Col( PP^{T} Q V_{S})= (\prod_{i \in S} \sigma_{i}) Col( P U_{S}) \nl 
\Pi_{Q}^{k} \ket{Col(PV_{S})}  = Col( QQ^{T} P V_{S})= (\prod_{i \in S} \sigma_{i}) Col( Q U_{S}). 
}  
These relations show that $( \ket{Col(PU_{S})},  \ket{Col(QV_{S})} )$ is a pair of principal vectors for Jordan's lemma and that angles between the vectors is given by 
$\cos(\theta_{S})= (\prod_{i \in S} \sigma_{i})$.

\end{proof} 
\noindent We next present a quantum algorithm generalizing the quantum singular value estimation algorithm \cite{KP16} to a subspace analog of SVD, that is given a matrix $A$ and a subspace state we can decompose the subspace 
into individual components corresponding to their projections on the dimension $d$ spaces of the singular vectors of $A$. 

\begin{algorithm} [H]
\caption{Subspace SVE algorithm.} \label{SVE} 
\begin{algorithmic}[1]
\REQUIRE Embedding of $A \in \R^{m \times n}$ with singular value decomposition $A= U \Sigma V^{T}$ in an orthogonal matrix $U=(P^{T}Q) \in \R^{N \times N}$. Givens circuits for applying $P$ and $Q$. 
\REQUIRE State $\ket{\phi_{k}} =  \sum_{|S|=k} \alpha_{S} \ket{Col(V_{S}) }$, a linear combination of subspace states corresponding to subspaces spanned by sets of $k$ singular vectors os $A$. 
\ENSURE State $\sum_{S \in \mathcal{H}_{k}} \alpha_{S} \ket{Col(V_{S}) } \ket{ \overline{\theta_{S}} }$ s.t. $|\cos(\overline{\theta_{S}}) - \prod_{i \in S} \sigma_{i} | = O(\epsilon)$. 
\STATE Apply Givens circuit $G(Q)$ to obtain subspace state $\ket{Col(QX)}$ of dimension $N$. 
\STATE Apply the phase estimation algorithm for $\mathcal{G}(U)$ on  $\ket{Col(QX)}$ to obtain additive error $\epsilon$ estimates of the eigenvalues to obtain the quantum state
$ \sum_{S \in \mathcal{H}_{k}} \alpha_{S} \ket{Col(V_{S}) } \ket{ \overline{\theta_{S}} }$. 
\end{algorithmic}
\end{algorithm}
\noindent The algorithm stated above is the analog of singular value estimation for subspace states. The result can be formally stated as, 
\begin{theorem} 
(Subspace SVE) Let $A=U\Sigma V^{T}$ be the singular value decomposition for matrix $A \in \R^{m \times n}$ embedding in unitary $U \R^{N \times N}$. Then the mapping, 
\al{ 
\ket{\phi_{k}} := \sum_{S \in \mathcal{H}_{k}} \alpha_{S} \ket{Col(V_{S}) } \to  \sum_{S \in \mathcal{H}_{k}} \alpha_{S} \ket{Col(V_{S}) } \ket{ \overline{\theta_{S}} }
} 
such that $|\cos(\overline{\theta_{S}}) - \prod_{i \in S} \sigma_{i} | = O(\epsilon)$ has complexity $O(T(U)/\epsilon)$ where $T(U)$ is the complexity of implementing the Givens circuits for $U$. 
\end{theorem}

The Givens circuits can be also used in a black box manner to implement singular value transformations
for projection and matrix inversion using known techniques \cite{gilyen2020improved}. Similar to the application of the original SVE algorithm to recommender systems, we believe these techniques can have applications to collaborative filtering based recommender systems where recommendations are made jointly to groups of $k$ users. Finding more quantum and quantum-inspired machine learning applications for compound matrix $SVD$ and subspace SVE remains a direction for future work.

\section{Topological data analysis} 
Topological data analysis is an area closely related to the constructions presented in this work. In this section, we show that the Clifford loaders introduced in section \ref{cliffload} can be used to reduce exponentially the depth for the best known quantum algorithm for topological data analysis. The result relies on the observation that the Dirac operator of an arbitrary simplicial complex embeds into the unitary matrix $\Gamma(x)$ for 
the vector $x= \frac{1}{ \sqrt{n}}  1^{n}$. The Clifford loader constructions thus provide a logarithmic depth implementation of a block encoding for the Dirac operator which can be utilized for TDA. The resulting algorithm 
reduces the depth for the topological data analysis algorithm from $O(n)$ to $O(\log n)$. 

We note that the circuit depth for the original LGZ algorithm \cite{lloyd2016quantum} for quantum TDA was $O(n^{5})$, a series of recent works \cite{gyurik2020towards, hayakawa2021quantum} on quantum TDA have improved on this by more efficient encodings for the Dirac operator and brought down the depth to $O(n)$ \cite{ubaru2021quantum}. Using logarithmic depth Clifford loaders to implement a block encoding for the Dirac operator provides an exponential improvement in terms of the circuit depth over these constructions as well.

We next recall the definitions of a simplicial complex and the Laplacian and the Dirac operator for a simpliclal complex and prove the main result on the embedding of the Dirac operator into the Clifford loader unitary. 
\begin{defn} 
A rank $r$ simplicial complex is a downward closed collection of subsets $V_{k} \subset \mathcal{H}_{k}$ for $0\leq k \leq r$. The elements of $V_{k}$ are referred to as $k$ simplices. 
\end{defn} 
\noindent We consider two specific examples of simplicial complexes relevant to machine learning. First, the Vietoris Rips complex $VR(X, d)$ of a set of points $X$ in a metric space consists of all simplices having diameter at most $d$. A second example is graph-based simplicial complex where $V_{k}$ be the set of $k$-cliques contained in a graph $G(V,E)$. 

An orientation is a choice of permutation (an ordering) for each simplex, a simplicial complex is said to be orientable if there is a consistent choice of orientation under intersections. A triangulation of a Mobius strip is an example of a non-orientable complex, we will work with orientable complexes in this work. A choice of orientation is needed to explicitly write the matrix representation of the Dirac operator, but the spectral properties of the Dirac operator are 
independent of the orientation. The orientation is a permutation on the vertices of the simplex. 
\begin{defn} \label{orient}
If a $k$-simplex has orientation $\sigma$, then the sub-simplex with 
the $\sigma_j$ deleted has induced orientation $(-1)^{j-1} \sigma^j$ where $\sigma^j$ is the permutation obtained by deleting $\sigma_{j}$ from $\sigma$. 
\end{defn}
\noindent Note that $(-1)^{j-1}$ counts the parity of the number of swaps to move $\sigma_j$ to $\sigma_1$. Fix a global orientation so that the vertices of all the simplices are in increasing order. Let $x^j$ be the simplex obtained from $x$ by deleting the $j$-th element in this order. The explicit description of the non-zero entries of the Dirac operator is as follows,  
\all{ 
d(x,x^j) &= (-1)^{j-1}  \hspace{30pt} \text{$\forall j\in [k]$} \nl
D(x,y) &= d(x,y) + d^{*}(x,y)
} {one} 
\noindent Note that $D$ is an exponentially large sparse matrix, a fact that has been used to obatin for quantum linear algebra algorithms for working with $d$. Let us establish some further properties of the Dirac operator.

\begin{claim} 
The matrix $d$ is a lower triangular matrix such that $d^{2}=0$. 
\end{claim} 
\begin{proof} 
The $d$ matrix is lower triangular as multiplication by $d$ reduces the number of elements in each simplex by $1$, similarly $d^{*}$ is upper triangular. 
Let $x$ be a simplex, then $d^{2} x$  is a sum over sub-simplices of $x$ with two vertices deleted and these occur with opposite signs (it is a basic fact 
that according to Definition \ref{orient} , the orientations for the simplices $x^{ij}$ and $x^{ji}$ are different), implying that $d^{2}=0$. 
\end{proof} 
\noindent The Dirac operator $D= (d + d^{*})$ squares to the Laplacian $\Delta = dd^{*} + d^{*} d$.  
\begin{claim} 
The Laplacian is a block diagonal matrix with diagonal entries $L_{p}(xx)= p+ 1+ext_{p}(x)$ with $ext_{p}(x)$ being the 
number of extensions of $p$-simplex $x$ to a $(p+1)$ simplex. 
\end{claim} 
\begin{proof} 
Let us compute the entries of the Laplacian by computing explicitly the entries of $d_{p} d_{p}^{*}$ and  $d_{p}^{*} d_{p}$. Then, $L_{p}(xx)= p+ 1+ |ext(x)|$ as there are $p+1$ simplices of size $p-1$ in $x$. For the off diagonal entries we have $L_{p} (xy) = \pm 1$ if $|x \cap y|= p-1$ and $x \cup y$ 
is not a simplex and $L_{p} (xy) = 0, \pm 2$ if $|x \cap y|= p-1$ and $x \cup y$ is a simplex. 
\end{proof} 

We give a representation for the Dirac operator for the complete simplex and shoe that is a scaled unitary operator as it squares to $nI$. The representation is used later  
to design efficient quantum circuits for implementing the Dirac operator. 

\begin{lemma} 
The Dirac operator $D$ for the complete simplicial complex on $n$ vertices can be represented as, 
\al{ 
D =  \sum_{i \in [n]} Z^{\otimes (i-1)}\otimes  X \otimes  I^{\otimes (n-1)}.
} 
\end{lemma} 
\begin{proof} 
We verify that the rows of $D$ and for the operator $\sum_{i \in [n]} Z^{\otimes (i-1)} \otimes X \otimes I^{\otimes (n-1)}$ are the same for all simplices $x$, that is $\forall x \subseteq [n]$ 
\al{ 
D(x, .) = X \ket{x} + Z\otimes X \ket{x} + \cdots + Z^{\otimes (n-1)} \otimes X \ket{x}. 
} 
In order to verify this note that (i) Let $k$ be the $j$-th index in sorted order such that $x_{k}=1$, then $Z^{\otimes (k-1)} \otimes X \ket{x} = (-1)^{j-1} \ket{x^{j}}= d(x, x^{k})$ as the $(j-1)$ indices less than 
$k$ where $x_{j}=1$ each contribute a $(-1)$. (ii) Let $k$ be an index such that $x_{k}=0$, then $ Z^{\otimes (k-1)} \otimes X \ket{x} = (-1)^{|i< k| x_{i}=1|} \ket{x^k}= d(x^{k}, x)= d^{*} (x, x^{k})$. It follows that the Dirac operator for the complete simplex is $\sum_{i \in [n]} Z^{\otimes (i-1)} \otimes X \otimes I^{\otimes (n-1)}$. 
\end{proof} 

\noindent The Dirac operator $\frac{1}{\sqrt{n}} \sum_{i \in [n]} Z^{\otimes (i-1)}\otimes  X \otimes I^{\otimes (n-1)}$ for the complete simplicial complex is unitary. In fact it also has a logarithmic depth implementation as it is identical to the unitary $\Gamma(x)$ implemented by the Clifford loader circuit  for the vector $x= \frac{1}{\sqrt{n}} 1^{n}$ with uniform coordinates.
\begin{theorem} 
There is a quantum circuit with $n$ qubits, $O(n \log n)$ gates and depth $O(\log n)$ for the Dirac operator $\frac{1}{\sqrt{n}} \sum_{i \in [n]} Z^{\otimes (i-1)}\otimes  X \otimes I^{\otimes (n-1)}$ for the complete simplicial complex. 
\end{theorem}

\noindent For the case of a more general simplicial complex, the Dirac operator is no longer a scaled unitary. However, the Dirac operator is a scaled submatrix of the unitary U
as we show in the next claim.  
\begin{theorem} \label{enc} 
The Dirac operator of an orientable simplicial complex $C$ is the submatrix of  $U= \frac{1}{\sqrt{n}} \sum_{i \in [n]} Z^{\otimes (i-1)}\otimes  X \otimes I^{\otimes (n-1)}$ indexed by the simplices that belong to $C$. 
\end{theorem} 
\begin{proof} 
As the entries $D(x,y)=0$ unless $x, y \in C$, it suffices to show that $D(x,y)$ agrees with the restriction of $U$ to the vertices in $C$. 
The non-zero entries of $d(x, x^j)$ in equation \eqref{one} belong to the restricted sub-matrix as $x^{j} \in C$ by downward closure.
If there is an orientation on $C$ compatible with the induced orientations for the complete complex then $d(x, x^j)$ has the same sign for both $C$ 
and the complete complex, establishing the result. 
\end{proof} 
\noindent The Dirac operator $U$ for the complete simplex thus provides a unitary block encoding for the Dirac operator of a general simplicial complex $C$. 
Theorem \ref{enc} shows that the operator $\sqrt{n} U= \sum_{i \in [n]} Z^{\otimes (i-1)}\otimes  X \otimes I^{\otimes (n-1)}$ is a block encoding for the Dirac operator of 
a simplicial complex. 

For TDA, we do not need the full power of a quantum linear system solver, we only need to be able to perform eigenvalue estimation up to precision proportional 
to the spectral gap for the Dirac operator, the cost for the eigenvalue estimation using the qubitization technique is $O( \frac{\sqrt{n}}{\norm{D}} \kappa \log (n\kappa/\epsilon)(T_{U} + T_{O}))$ where $T_{U}+T_{O}$ is the time needed to implement the block encoding for the Dirac operator. 

The detailed running time and the assumptions required for exponential speedups for state of the art quantum topological data analysis algorithms may be found for example in Theorem 7 in \cite{hayakawa2021quantum}. Our contribution to quantum TDA is to provide an exponentially lower-depth block encoding for the Dirac operator by embedding it into a logarithmic depth unitary. The embedding can be used in a black box way and combined with the Dicke state preparation and singular value transformation routines used in the state-of-the-art quantum TDA algorithms to achieve an exponential depth reduction for these approaches.

\bibliographystyle{IEEEtranS} 
\bibliography{b1}

\end{document}